\documentclass[11pt]{article}
\usepackage{amsmath, amssymb, amsthm, graphicx}
\usepackage[makeroom]{cancel}
\usepackage{tcolorbox}
\usepackage[colorlinks=true,linkcolor=blue,citecolor=blue]{hyperref}
\usepackage{tikz-cd,mathtools}
\usepackage{array}
\usepackage{proof}
\usepackage{stmaryrd}
\usepackage{subcaption}

\newcommand{\up}[2][]{\W_{#1}(#2)}

\newcommand{\I}{\mathbb{I}}
\newcommand{\II}{\mathcal{I}}
\newcommand{\J}{\mathbb{J}}

\renewcommand{\O}{\mathcal{O}}

\newcommand{\N}{\mathbb{N}}
\newcommand{\W}{\mathcal{W}}
\newcommand{\Z}{\mathbb{Z}}

\newcommand{\id}{\mathrm{id}}
\newcommand{\ins}[2]{\mathrm{ins}_{#1}(#2)}
\newcommand{\del}[2]{\mathrm{del}_{#1}(#2)}
\newcommand{\comp}[1]{\langle#1\rangle}
\newcommand{\ma}{maj}

\newcommand{\dom}{\mathrm{dom}}

\newcommand{\floor}[1]{\lfloor #1 \rfloor}
\newcommand{\bigfloor}[1]{\left\lfloor #1 \right\rfloor}
\newcommand{\bigceil}[1]{\left\lceil #1 \right\rceil}
\newcommand{\ceil}[1]{\lceil #1 \rceil}

\renewcommand{\int}[1]{{\llbracket #1\rrbracket}}

\newcommand{\bef}{\mathrel{\angle}}
\newcommand{\pic}[2][-.5mm]{\raisebox{#1}{\includegraphics{Images/#2}}}

\newcommand{\res}[1]{|_{#1}}
\newcommand{\im}[1]{\mathrm{im}(#1)}

\newcommand{\restr}[2]{#1\res{#2}}

\newcommand{\slot}{\raisebox{-.7mm}{\textcolor{gray}{\text{\Large$\blacksquare$}}}}

\newcommand{\maj}{\text{maj}}

\newcommand{\Mon}{\mathsf{Mon}}

\newcommand{\Di}{\mathrm{D}}

\newtheorem{theorem}{Theorem}[section]
\newtheorem*{theorem*}{Theorem}
\newtheorem{lemma}{Lemma}[section]
\newtheorem*{lemma*}{Lemma}
\newtheorem{corollary}{Corollary}[section]
\newtheorem{proposition}{Proposition}[section]
\newtheorem*{proposition*}{Proposition}

\theoremstyle{definition}
\newtheorem{definition}{Definition}[section]

\theoremstyle{remark}
\newtheorem{example}{Example}[section]
\newtheorem{remark}{Remark}[section]

\newtheorem{claim}{Claim}

\newcommand{\dlArrow}{\arrow[dl, start anchor=south west, end anchor=north east]}
\newcommand{\ulArrow}{\arrow[ul, dashed, start anchor=north west, end anchor=south east]}
\newcommand{\wid}[2][]{\mathmakebox[\widthof{$#1$}]{#2}}

\title{Local generation of languages: the monotonic binary sequences}
\author{Mathieu Hoyrup\footnote{Université de Lorraine, CNRS, Inria, Nancy, 54000, France}}
\begin{document}
\maketitle

\begin{abstract}
In a previous article, we have introduced the problem of local generation of languages, where the communication underlying the generation procedure is captured by a simplicial complex. We study in details this problem for the language of binary monotonic sequences. We prove general results and identify several classes of minimal simplicial complexes generating this language.
\end{abstract}

\tableofcontents
\section{Introduction}
For~$n\in\N$, we consider the language~$\Mon_n\subseteq \{0,1\}^n$ of binary sequences that are monotonic, i.e.~non-decreasing or non-increasing. For instance,
\[
\Mon_4=\{0000,0001,0011,0111,1111,1110,1100,1000\}.
\]

We investigate the problem of generating the strings in~$\Mon_n$ in a local way as follows. Each cell in~$I_n=\{0,\ldots,n-1\}$ has to produce a bit in such a way that they collectively produce any string in~$\Mon_n$. To achieve this goal, the cells need to communicate, and we analyze how much communication is needed between them.

Let us formulate this problem precisely, using the framework introduced in \cite{H25}. A language is a set~$L\subseteq A^I$ for some finite sets~$A$ and~$I$, and a generation procedure for~$L$ is simply a function from strings to strings whose image is~$L$ (the input strings are elements of~$B^J$ for some finite sets~$B$ and~$J$). Such a function~$f:B^J\to A^I$ has an intrinsic communication structure: output cells in~$I$ implicitly communicate by reading the contents of common input cells in~$J$. This structure is captured by the communication complex~$K_f$ of~$f$, whose vertex set is~$I$ and whose simplices are the subsets of~$I$ that read a common input cell. We then say that a simplicial complex over~$I$ generates~$L$ if there is a function~$f$ generating~$L$ such that~$K_f\subseteq K$. The general problem is then to identify the simplicial complexes that generate a given language~$L$, and it is sufficient to identify the minimal ones.

We apply this framework to~$\Mon_n$, whose behavior turns out to be particularly rich. In particular, there does not seem to be a uniform way of describing the minimal complexes generating~$\Mon_n$ when~$n$ grows.

\paragraph{Overview of the main results.}

Each way of generating~$\Mon_n$ can be extended into a generation procedure for~$\Mon_{n+1}$ in the simple following way: produce an element of~$\Mon_n$ and insert a bit at a given position. The allowed values for the new bit only depend on the values appearing at the two positions surrounding the insertion, so the new position only needs to communicate with these two ones. Therefore, every complex~$K$ generating~$\Mon_n$ induces complexes generating~$\Mon_{n+1}$, obtained by inserting a vertex in~$K$ in a certain way (Definition \ref{def_insertion}). The minimality of~$K$ does not always imply the minimality of the new complex, but we identify criteria to preserve minimality (Proposition \ref{prop_insertion_minimal2}).

Therefore, the minimal complexes generating~$\Mon_n$ fall in two categories: either they are obtained from minimal complexes generating~$\Mon_{n-1}$ by vertex insertions, or they pop up for this particular value of~$n$, providing a new way of generating~$\Mon_n$. For~$n=2$ one has~$\Mon_2=\{0,1\}^2$ so there is a trivial generation procedure where cells do not need to communicate at all, which is materialized by the complex~$K_2=\comp{\{0\},\{1\}}$ made of two vertices and no edge. We show that for~$n=5,7,8$, new generation procedures appear, materialized by complexes that we naturally call~$K_5,K_7$ and~$K_8$ (Propositions \ref{prop_K5}, \ref{prop_K7} and \ref{prop_K8}).

We then establish a partial classification of the minimal complexes generating~$\Mon_n$.

\begin{proposition*}[Proposition \ref{prop_mono_2_simplices}]
The minimal complexes generating~$\Mon_n$ and containing several simplices of size~$n-1$ are obtained from~$K_2$ by vertex insertions, and have the form~$\comp{I_n\setminus \{a\},I_n\setminus \{b\}}$ where~$a,b\in I_n$ are distinct.
\end{proposition*}


\begin{theorem*}[Theorem \ref{thm_one_interval}]
The minimal complexes generating~$\Mon_n$ and containing one simplex of size~$n-1$ are obtained from~$K_5$ by vertex insertions.
\end{theorem*}

In all these complexes, the maximal simplices turn out to be intervals, i.e.~sets of the form~$[i,j]$ or~$[0,i]\cup [j,n-1]$ for~$0\leq i<j<n$ (which are indeed intervals when identifying~$I_n$ with~$\Z/n\Z$ and arranging it on a circle). We show that it is not a coincidence.
\begin{theorem*}[Theorem \ref{thm_intervals}]
Let~$K$ be a minimal complex generating~$\Mon_n$. The maximal simplices of~$K$ are all intervals.
\end{theorem*}

We also prove a result that holds for all~$n$, studying how small the intervals of a complex generating~$\Mon_n$ can be. More precisely, let~$\mu(n)$ be the minimal~$k$ such that there is a complex generating~$\Mon_n$ whose intervals have lengths at most~$k$.  We show that~$\mu(n)\simeq \frac{3n}{4}$.
\begin{theorem*}[Theorem \ref{thm_mu}]
For~$n\geq 8$, one has
\[
\bigfloor{\frac{3n+1}{4}}\leq \mu(n)\leq \bigceil{\frac{3n}{4}}.
\]
\end{theorem*}
We also introduce an inference system that can, in principle, be used to express most of the arguments showing that a complex does not generate a language. This inference system has been implemented \cite{Prover} and used to help finding, by exhaustive search, some of the most technical arguments of the article. Note that this automatic tool is closer to a proof assistant than to an automatic prover, and is only used to test hypotheses and solve particular cases. It quickly becomes impractical as~$n$ grows, and cannot by itself find proofs of general statements such as Theorem \ref{thm_mu}.

The idea of performing computations in a local way with little communication between cells has been formalized and studied in many ways: cellular automata \cite{Kari05,DFP12}, automata networks \cite{GM90,Gad19,GGPT21}, distributed reactive systems \cite{PnueliR90}, distributed networks \cite{NS95}, distributed graph automata \cite{R15}, distributed environments \cite{FG18}, distributed computing \cite{HS99,HKR13}. It is customary to use graphs or simplicial complexes to represent the communication structure in a distributed setting: the interaction graph \cite{Gad19} in automata networks, or the input and output complexes in distributed computing \cite{HS99}. As explained in \cite{H25}, the local generation problem investigated in the present article can be reformulated using combinatorial topology as in \cite{HS99,HKR13}, and we briefly discuss this reformulation in the case of~$\Mon_n$ at the end of the article.


The article is organized as follows. In Section \ref{sec_def}, we briefly present the main definitions. In Section \ref{sec_mon_examples}, we identify some complexes generating~$\Mon_n$ for~$n=5,7,8$. In Section \ref{sec_general} we prove structural results, for instance the fact that minimal complexes generating~$\Mon_n$ are always made of intervals. In Section \ref{sec_rule_system} we present a unifying framework for presenting proofs of negative results, as well as technical tools that are used in the next arguments. In Section \ref{sec_families}, we prove the minimality of the families of complexes generating~$\Mon_n$ presented in Section \ref{sec_mon_examples}. In Section \ref{sec_34}, we show that in a complex generating~$\Mon_n$, the maximal length of its intervals is at least approximately~$\frac{3n}{4}$, and that this lower bound is tight. In Section \ref{sec_future}, we briefly discuss the reformulation of the problem in terms of combinatorial topology and list a few open questions. We defer some technical proofs to Section \ref{sec_app}.

\section{Definitions}\label{sec_def}
\subsection{Background}\label{sec_back}
We briefly recall the main notions introduced in \cite{H25}. There are two main objects in this study: languages and communication complexes.

A \textbf{language} is a set~$L\subseteq A^I$, where~$A,I$ are finite sets. We often consider the case~$I=I_n=[0,n-1]$ for some~$n\geq 1$, and denote~$A^{I_n}$ by~$A^n$. If~$B,J$ are finite sets, then a function~$f:B^J\to A^I$ \textbf{generates}~$L$ if~$L=\im{f}$. The elements of~$I$ are called \textbf{output cells}, the elements of~$J$ are called \textbf{input cells}. Each output cell~$i\in I$ evaluates its own function~$f_i(x)$, which is the value of~$f(x)$ at position~$i$. Usually, it only depends on the values of the input~$x$ at certain input cells. In order to produce the strings of~$L$ collectively, the output cells need some level of coordination, which is made possible by the implicit communication that is allowed by reading common input cells. The communication complex of~$f$ captures this underlying communication structure. To each~$i\in I$ is associated its \textbf{input window} $\W_f(i)\subseteq J$, which is the smallest subset~$W\subseteq J$ which determines the value of~$f_i$, i.e.~such that for all~$x,y\in B^J$, if their restrictions to~$W$ coincide, written~$\restr{x}{W}=\restr{y}{W}$, then~$f_i(x)=f_i(y)$. The dual windows are defined for~$j\in J$ as~$\W^f(j)=\{i\in I:j\in \W_f(i)\}$. The \textbf{visibility diagram} of~$f$ is the set~$\{(i,j)\in I\times J:j\in W_f(i)\}=\{(i,j)\in I\times J:i\in W^f(j)\}$, represented as a binary-valued matrix. For~$S\subseteq I$ we define
\[
\up[f]{S}=\bigcap_{i\in S}\W_f(i)=\{j\in J:S\subseteq \W^f(j)\},
\]
which is the set of input cells that are visible by \emph{all} the output cells in~$S$. The \textbf{communication complex} of~$f$, denoted by~$K_f$, is the simplicial complex over vertex set~$I$, defined for~$S\subseteq I$ by
\[
S\in K_f\iff \up[f]{S}\neq\emptyset.
\]
We will often write a simplicial complex~$K$ as~$K=\comp{S_0,\ldots,S_k}$ where the~$S_i$'s are the maximal simplices of~$K$.

\begin{definition}[Language generation]
Let~$L\subseteq A^I$ be a language and~$K$ a simplicial complex over vertex set~$I$. We say that~$K$ \textbf{generates}~$L$ if there exist~$B,J$ and a function~$f:B^J\to A^I$ such that~$\im{f}=L$ and~$K_f\subseteq K$.
\end{definition}

The general goal of the study is to describe, for any given language~$L$, the set of all the simplicial complexes generating~$L$. This set is upwards closed, i.e.~if~$K\subseteq K'$ and~$K$ generates~$L$, then~$K'$ generates~$L$ as well, so our goal is to identify the minimal complexes generating~$L$.

We list a few general results from \cite{H25} that will be used in this article. First, if a complex generates a language, then one can choose the input space of the generating function in a canonical way.
\begin{proposition}[Canonical form]\label{prop_canonical}
Let~$L\subseteq A^I$ and~$K$ be a complex over vertex set~$I$ that generates~$L$. Let~$B=L$ and let~$J$ be the set of maximal simplices of~$K$. There exists a function~$f:B^J\to A^I$ generating~$L$, such that~$K_f\subseteq K$ and for each~$x\in L$,~$f(x,\ldots,x)=x$.
\end{proposition}
In particular, there is always a generation procedure working as follows: each maximal simplex is assigned a value in~$B$ and each vertex~$i\in I$ has a local rule that, upon reading the values of its incident simplices, determines an output value in~$A$ to be assigned to~$i$. Moreover, one can choose~$B=L$ and if all the simplices that are incident to a vertex~$i$ have the same value~$x\in L$, then~$v$ takes value~$x_i$.

If a language~$M$ is the image of a language~$L$ by some function~$f$, then any generation function for~$L$ induces a generation function for~$M$ by composition. The way it transforms the communication complex can be easily expressed in terms of~$f$ as follows.
\begin{proposition}[Image of a language]\label{prop_image}
Let~$L\subseteq {A_0}^{I_0}$ and~$f:{A_0}^{I_0}\to {A_1}^{I_1}$. If a complex~$K$ generates~$L$, then the complex~$f_*(K)$ generates~$f(L)$, where~$f_*(K)$ is induced by the following sets, for~$S\in K$:
\[
f_*(S)=\bigcup_{j\in S}\W^f(S).
\]
\end{proposition}

Let~$L\subseteq A^I$. We say that two distinct positions~$i,j\in I$ are \textbf{independent} w.r.t.~$L$ if for all~$a,b\in A$ such that there exist~$x,y\in L$ satisfying~$x_i=a$ and~$y_j=b$, there exists~$z\in L$ satisfying~$z_i=a$ and~$z_i=b$. Intuitively, when generating a string in~$L$, the values at these positions can be chosen independently, i.e.~with no communication between~$i$ and~$j$. Indeed, as proved in \cite[Corollary 2.4]{H25}, if~$i$ and~$j$ are independent w.r.t.~$L$, then the complex~$K=\comp{I\setminus \{i\},I\setminus \{j\}}$ generates~$L$. The idea is simple:~$i$ and~$j$ directly take their values from two distinct input cells that are visible by all the other output cells; those ones therefore know the values taken by~$i$ and~$j$ so they can collectively agree on some sequence in~$L$ extending these values.

In this article, we analyze the language~$\Mon_n\subseteq\{0,1\}^n$ which is the set of monotonic (i.e.~non-decreasing or non-increasing) binary sequences. Concretely,
\[
\Mon_n=\{0^{n-k}1^k:0\leq k< n\}\cup \{1^{n-k}0^{k}:0\leq k< n\}.
\]
Our goal is to identify all the simplicial complexes generating~$\Mon_n$ for any value of~$n$.

\subsection{Intervals}
It will be convenient to identify~$I_n=[0,n-1]$ with~$\Z/n\Z$, which can be graphically arranged in a circle. As we will see in Section \ref{sec_intervals}, the analysis of~$\Mon_n$ reveals that the simplices made of ``consecutive'' positions, that we call \emph{intervals}, play an important role.

More precisely, for~$a,b\in\Z$ we define~$\int{a,b}\subseteq I_n$ as follows:
\[
\int{a,b}=\{c\bmod n:c\in\Z, a\leq c\leq b'\},
\]
where~$b'\in\Z$ is the unique representative of~$b$ modulo~$n$ satisfying~$a\leq b<a+n$. An \textbf{interval} is~$\int{a,b}$ for some~$a,b\in\Z$. Note that~$\int{a,a+n-1}=I$ but~$\int{a,a+n}=\int{a,a}=\{a\}$. As the interval~$\int{a,b}$ only depends on the equivalence classes of~$a$ and~$b$ modulo~$n$, we will sometimes use the notation~$\int{a,b}$ where~$a,b\in\Z/n\Z$.

The size of~$\int{a,b}$ is~$1+((b-a)\bmod n)$, which is the representative modulo~$n$ of~$1+b-a$ in~$[1,n]$. For~$k\in [1,n]$, a \textbf{$k$-interval} is an interval of size~$k$.

Let~$\I,\J$ be two intervals. We write~$\I\bef\J$ if~$\I\cup \J=\int{a,d}$ where~$\I=\int{a,c}$ and~$\J=\int{b,d}$. 
Intuitively, it means that~$\I$ and~$\J$ overlap and~$\I$ is located ``before’’~$\J$. We give a simple criterion, illustrated in Figure \ref{fig_before}.
\begin{figure}[ht]
\centering
\includegraphics{Images/intervals-0}
\caption{$\I\bef\J$}\label{fig_before}
\end{figure}

\begin{lemma}\label{lem_bef}
Let~$a,b,c,d\in \Z$. If~$a\leq b\leq c\leq d<a+n$, then~$\int{a,c}\bef\int{b,d}$.
\end{lemma}
\begin{proof}
As~$a\leq b\leq c\leq d$, the intervals in~$\Z$ satisfy~$[a,c]\cup [b,d]=[a,d]$.

Let~$f:\Z\to\Z/n\Z$ be the function sending a number to its equivalence class modulo~$n$. For a pair~$(x,y)\in\Z^2$, if~$x\leq y<x+n$, then~$f([x,y])=\int{x,y}$.

The conditions of the lemma imply that the pairs~$(a,c)$,~$(b,d)$ and~$(a,d)$ satisfy these inequalities, therefore~$\int{a,c}\cup \int{b,d}=f([a,c])\cup f([b,d])=f([a,c]\cup [b,d])=f([a,d])=\int{a,d}$.
\end{proof}

When using Lemma \ref{lem_bef} to show that~$\int{a,c}\bef\int{b,d}$, we will summarize the inequalities that need to be checked by the following diagram:
\begin{center}
\begin{tikzcd}
a\arrow[r]&b\dlArrow\\
c\arrow[r]&d\ulArrow
\end{tikzcd}
\end{center}
where a plain arrow~$a\to b$ means~$a\leq b$ and a dashed arrow~$d\dashrightarrow a$ means~$d<a+n$. An advantage of these diagrams will become apparent in the proofs of Lemmas \ref{lem_technical} and \ref{lem_b}, because they can be concatenated and summarize a large quantity of inequalities in a condensed way.
\subsection{Graphical representation of complexes}
The simplicial complexes encountered in this study usually have high dimensions, so can hardly be visualized. Fortunately, all the simplicial complexes considered in this article will be generating by intervals, and therefore have a simple graphical representation. We illustrate this representation on the next example, that we will meet again in Section \ref{sec_K5}.

Let~$K_5$ be the simplicial complex over~$I_5=[0,4]$ be defined by
\[
K_5=\comp{\int{0,2},\int{1,3},\int{2,4},\int{3,1}}.
\]
This complex can be visualized as shown in Figure \ref{fig_K50}, in which the first column lists the vertices in~$[0,n-1]$ and the next columns represent the maximal intervals. This representation is closely related to the visibility diagram of a function, but is more condensed because it only shows the maximal simplices, and shows each of them exactly once.
\begin{figure}[ht]
\centering
\includegraphics{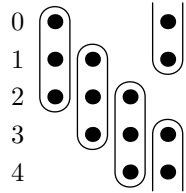}
\caption{A graphical representation of~$K_5$
}\label{fig_K50}
\end{figure}

\subsection{Symmetries}
The language~$\Mon_n$ has two types of symmetries. For~$b\in\{0,1\}$, let~$\overline{b}=1-b$. If~$x=x_0\ldots x_{n-1}$ is monotonic, then:
\begin{itemize}
\item Its reflection~$x_{n-1}\ldots x_0$ is monotonic,
\item The sequence~$\overline{x_{n-1}}x_0\ldots x_{n-2}$, obtained by applying a circular permutation and a bit flip, is monotonic.
\end{itemize}

These symmetries are captured by the action of the dihedral group
\[
\Di_{2n}=\left\langle r,s|r^{2n}=s^2=(sr)^2=1\right\rangle
\]
on~$\{0,1\}^n$, defined as follows:
\begin{align*}
r\cdot (x_0\ldots x_{n-1})&=\overline{x_{n-1}} x_0 \ldots x_{n-2},\\
s\cdot (x_0\ldots x_{n-1})&=x_{n-1}\ldots x_0.
\end{align*}

The action of~$\Di_{2n}$ on~$\Mon_n$ is transitive, i.e.~if~$x,y\in\Mon_n$ then there exists~$g\in \Di_{2n}$ such that~$y=g\cdot x$. Indeed, the orbit of~$0^n$ under the action of~$r$ visits every element of~$\Mon_n$:
\newcommand{\tor}{\to}
\[
0^n\tor 10^{n-1}\tor \ldots\tor 1^{n-1}0\tor 1^n\tor 01^{n-1}\tor \ldots\tor 0^{n-1}1\tor 0^{n}.
\]
We will use transitivity in some arguments, proving a result for one sequence and obtaining the full result by symmetry.

The symmetries of~$\Mon_n$ induce symmetries on the class of simplicial complexes generating~$\Mon_n$. To state it precisely, we consider the action of~$\Di_{2n}$ on~$I_n$:
\begin{align*}
r\cdot i&=i+1\bmod n,\\
s\cdot i&=n-1-i.
\end{align*}
The two actions are closely related:~$(g\cdot x)_i$ is a function of~$g,i$ and~$x_{g^{-1}\cdot i}$, i.e.~does not depend on the values of~$x$ at other positions.
 
If~$K$ is a simplicial complex over~$I_n$ and~$g\in \Di_{2n}$, then let~$g\cdot K=\{g\cdot S:S\in K\}$ where~$g\cdot S=\{g\cdot i:i\in S\}$ for~$S\subseteq I_n$.
\begin{proposition}[Symmetries]\label{prop_mon_symmetries}
If a complex~$K$ generates~$\Mon_n$, then for each~$g\in \Di_{2n}$, the complex~$g\cdot K$ generates~$\Mon_n$.
\end{proposition}
\begin{proof}
Let~$g\in \Di_{2n}$ and consider the function~$f:x\mapsto g\cdot x$. In order to evaluate~$f_i(x)$, one only needs to know the value of~$x$ at position~$g^{-1}\cdot i$: it is either~$x_{g^{-1}\cdot i}$ or~$\overline{x_{g^{-1}\cdot i}}$ depending on~$g$ and~$i$. Therefore, one has~$\W_f(i)=\{g^{-1}\cdot i\}$, hence~$\W^f(j)=\{g\cdot j\}$ and~$f_*(K)=g\cdot K$ for any complex~$K$. Proposition \ref{prop_image} implies that if~$K$ generates~$\Mon_n$, then~$f_*(K)=g\cdot K$ generates~$f(\Mon_n)=\Mon_n$.
\end{proof}

A similar result was proved in \cite[Proposition 2.6]{H25}, assuming that the symmetry group acts on the language by permuting the symbols. Here, the action also flips the values of the symbols, but the argument is essentially the same.
%

\section{Examples of complexes generating \texorpdfstring{$\Mon_n$}{Mon}}\label{sec_mon_examples}
In this section, we identify certain complexes generating~$\Mon_n$, and will later prove their minimality. These examples are particular cases and we currently do not have a complete classification of the minimal complexes generating~$\Mon_n$ for any~$n\in\N$.

The first observation is that for~$n\geq 2$, any pair of distinct elements~$i,j\in I_n$ are independent w.r.t.~$\Mon_n$, because any assignment of binary values to~$i$ and~$j$ can be extended into a monotonic sequence. As explained in Section \ref{sec_back}, it implies that the complex~$\comp{I_n\setminus\{i\},I_n\setminus\{j\}}$ generates~$\Mon_n$. 
Once we have developed the necessary machinery, we will be able to prove that this complex is minimal.

It turns out that for~$n\geq 5$ there are other complexes generating~$\Mon_n$. We give three examples, for~$n=5,7$ and~$8$. The corresponding generation procedures were found by first rejecting other complexes and then designing a corresponding generation procedure by trial and error. We will see later that these complexes are minimal generating~$\Mon_n$.

\subsection{A complex generating \texorpdfstring{$\Mon_5$}{Mon5}}\label{sec_K5}
For~$n=5$, let
\[
K_5=\comp{\int{0,2},\int{1,3},\int{2,4},\int{3,1}}
\]
be the complex that was already shown in Figure \ref{fig_K50}, and that we repeat here in Figure \ref{fig_K5}. We show that it generates~$\Mon_5$.
\begin{figure}[ht]
\centering
\includegraphics{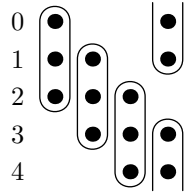}
\caption{The complex~$K_5$
}\label{fig_K5}
\end{figure}
\begin{proposition}\label{prop_K5}
The complex~$K_5$ generates~$\Mon_5$.
\end{proposition}
\begin{proof}
It is convenient to rename the elements of~$[0,4]$ as~$A,B,C,D,E$. We design a generation procedure which is a function~$f:\{0,1\}^{\{a,b,c,d,e\}}\to \{0,1\}^{\{A,B,C,D,E\}}$ that takes a binary sequence as input and corrects it to make it monotonic. Its visibility diagram is shown in Figure \ref{fig_vis_K5} (a square in row~$A$ and column~$a$ indicates that~$A$ reads~$a$, i.e.~$a\in\W_f(A)$ or equivalently~$A\in\W^f(a)$).

\begin{figure}[ht]
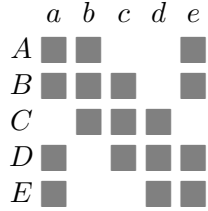

\centering
\begingroup
\arraycolsep=1pt
\begin{equation*}
\begin{array}{cccccc}
&a&b&c&d&e\\
A&\slot&\slot&&&\slot\\
B&\slot&\slot&\slot&&\slot\\
C&&\slot&\slot&\slot&\\
D&\slot&&\slot&\slot&\slot\\
E&\slot&&&\slot&\slot
\end{array}
\end{equation*}
\endgroup
\caption{The visibility diagram of a function generating~$\Mon_5$}\label{fig_vis_K5}
\end{figure}

The function~$f$ uses the ternary majority function~$\maj:\{0,1\}^3\to \{0,1\}$, which corrects the lack of monotonicity of its input, in the sense that the function
\begin{align*}
\{0,1\}^3&\to\{0,1\}^3\\
(x,y,z)&\mapsto (x,\maj(x,y,z),z)
\end{align*}
always outputs a monotonic sequence, and leaves the sequence unchanged if it is already monotonic. Observe that~$\maj$ commutes with bit flips:
\[
\maj(\overline{x},\overline{y},\overline{z})=\overline{\maj(x,y,z)}.
\]



The output cells have the following local rules: 
\begin{align*}
\mathtt{A(e,a,b)}&\mathtt{=\ma(\overline{e},a,b)}\\
\mathtt{B(e,a,b,c)}&\mathtt{=\ma(A,b,c)}\\
\mathtt{C(b,c,d)}&\mathtt{=\ma(b,c,d)}\\
\mathtt{D(c,d,e,a)}&\mathtt{=\ma(c,d,E)}\\
\mathtt{E(d,e,a)}&\mathtt{=\ma(d,e,\overline{a})},
\end{align*}
where~$\mathtt{\ma(A,b,c)}$ means~$\mathtt{\ma(A(e,a,b),b,c)=\ma(\ma(\overline{e},a,b),b,c)}$, and similarly for~$\mathtt{\ma(c,d,E)}$. Intuitively, each output cells reads the corresponding input cell as well as its neighbors and applies a majority vote, with the subtlety that~$B$ uses the output value~$A$ rather than the input value~$a$, and~$D$ uses the output value~$E$ rather than the input value~$e$.

\paragraph{Verification.} Let~$f:\{0,1\}^{\{a,b,c,d,e\}}\to \{0,1\}^{\{A,B,C,D,E\}}$ be the global function derived from these local rules. Its communication complex is~$K_5$ by construction, we check that its image is~$\Mon_5$.

First, if~$abcde$ is monotonic, then~$f(abcde)=abcde$, so~$\im{f}$ contains~$\Mon_5$. Conversely, we show that every output of~$f$ is monotonic. As~$f$ commutes with bit flips, i.e.~$f(\overline{abcde})=\overline{f(abcde)}$, it is sufficient to show the result when~$a=0$. A partial evaluation of~$f$ is shown in the next table, assuming~$a=0$:
%

\[
\begin{array}{l|l}
\mathtt{e=0}&\mathtt{e=1}\\
\hline
\mathtt{A=b}&\mathtt{A=0}\\
\mathtt{B=b}&\mathtt{B=b\land c}\\
\mathtt{C=\ma(b,c,d)}&\mathtt{C=\ma(b,c,d)}\\
\mathtt{D=d}&\mathtt{D=c\lor d}\\
\mathtt{E=d}&\mathtt{E=1}
\end{array}
\]
One easily checks that~$ABCDE$ is always monotonic.
\end{proof}

\subsection{A complex generating \texorpdfstring{$\Mon_7$}{Mon7}}

For~$n=7$, the complex
\[
K_7=\{\int{0,4},\int{1,5},\int{2,6},\int{4,1},\int{5,2}\}
\]
shown in Figure \ref{fig_K7} generates~$\Mon_7$.
\begin{figure}[ht]
\centering
\includegraphics{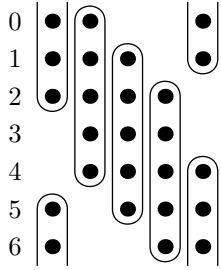}
\caption{The complex~$K_7$}
\label{fig_K7}
\end{figure}

\begin{proposition}\label{prop_K7}
The complex~$K_7$ generates~$\Mon_7$.
\end{proposition}
\begin{proof}
We define a function~$f:\{0,1\}^{\{a,b,d,f,g\}}\to \{0,1\}^{\{A,B,C,D,E,F,G\}}$, whose visibility diagram is shown in Figure \ref{fig_visibility_7}.
\begin{figure}[ht]
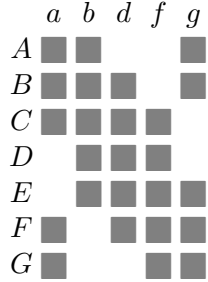

\begingroup
\arraycolsep=1pt
\begin{equation*}
\begin{array}{cccccc}
&a&b&d&f&g\\
A&\slot&\slot&&&\slot\\
B&\slot&\slot&\slot&&\slot\\
C&\slot&\slot&\slot&\slot&\\
D&&\slot&\slot&\slot&\\
E&&\slot&\slot&\slot&\slot\\
F&\slot&&\slot&\slot&\slot\\
G&\slot&&&\slot&\slot
\end{array}
\end{equation*}
\endgroup
\caption{The visibility diagram of a function generating~$\Mon_7$}\label{fig_visibility_7}
\end{figure}

As in Proposition \ref{prop_K5}, each output cell applies a majority vote, using a mixture of input and output values:
\begin{align*}
\mathtt{A(g,a,b)}&\mathtt{=\ma(\overline{g},a,b)}\\
\mathtt{B(g,a,b,d)}&\mathtt{=\ma(A,b,d)}\\
\mathtt{C(a,b,d,f)}&\mathtt{=\ma(a,b,D)}\\
\mathtt{D(b,d,f)}&\mathtt{=\ma(b,d,f)}\\
\mathtt{E(b,d,f,g)}&\mathtt{=\ma(D,f,g)}\\
\mathtt{F(d,f,g,a)}&\mathtt{=\ma(d,f,G)}\\
\mathtt{G(f,g,a)}&\mathtt{=\ma(f,g,\overline{a})}.
\end{align*}

\paragraph{Verification.}
We first show that the procedure only generates monotonic sequences. As it commutes with bit flips, we only need to consider the case~$a=0$. A partial evaluation of the function is presented in the next table. Note that~$\maj(x,y,0)=x\land y$ and~$\maj(x,y,1)=x\lor y$.  
\[
\begin{array}{l|l|l}
\text{If }\mathtt{g=1}&\multicolumn{2}{c}{\text{If }\mathtt{g=0}}\\
&\text{and }\mathtt{d=0}&\text{and }\mathtt{d=1}\\
\hline
\mathtt{A=0}&\mathtt{A=b}&\mathtt{A=b}\\
\mathtt{B=b\land d}&\mathtt{B=b}&\mathtt{B=b}\\
\mathtt{C=b\land (d\lor f)}&\mathtt{C=b\land f}&\mathtt{C=b}\\
\mathtt{D=\ma(b,d,f)}&\mathtt{D=b\land f}&\mathtt{D=b\lor f}\\
\mathtt{E=(b\land d)\lor f}&\mathtt{E=b\land f}&\mathtt{E=f}\\
\mathtt{F=d\lor f}&\mathtt{F=f}&\mathtt{F=f}\\
\mathtt{G=1}&\mathtt{G=f}&\mathtt{G=f}
\end{array}
\]
In any case, one easily checks that the output sequence is monotonic. Conversely, we show that every monotonic sequence is reached. As the function commutes with bit flips, it is sufficient to check that every sequence starting with~$\mathtt{A=0}$ is reached. The next table shows that every such sequence is the image of some input sequence. In order to check that the image of a particular input is a particular output, say $\mathtt{0001111}$, it is sufficient to check that~$\mathtt{CD=01}$ because then~$\mathtt{0001111}$ is its unique monotonic extension. In the next table, we underline the output values that need to be checked:
\begin{center}
\begin{tabular}{c|c}
\texttt{abdfg}&\texttt{ABCDEFG}\\
\hline
\texttt{00000}&\texttt{\underline{0}00000\underline{0}}\\
\texttt{00001}&\texttt{00000\underline{01}}\\
\texttt{00010}&\texttt{0000\underline{01}1}\\
\texttt{00011}&\texttt{000\underline{01}11}\\
\texttt{00110}&\texttt{00\underline{01}111}\\
\texttt{01011}&\texttt{0\underline{01}1111}\\
\texttt{01111}&\texttt{\underline{01}11111}
\end{tabular}
\end{center}
%
%
\end{proof}

\subsection{A complex generating \texorpdfstring{$\Mon_8$}{Mon8}}\label{sec_Mon8}

For~$n=8$, the complex
\[
K_8=\comp{\int{0,5},\int{2,7},\int{4,1},\int{6,3}}
\]
shown in Figure \ref{fig_K8} generates~$\Mon_8$.
\begin{figure}[ht]
\centering
\includegraphics{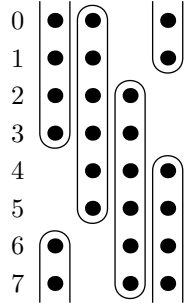}
\caption{The complex~$K_8$}
\label{fig_K8}
\end{figure}
\begin{proposition}\label{prop_K8}
The complex~$K_8$ generates~$\Mon_8$.
\end{proposition}

%
\begin{proof}
We design a function~$\phi:\{0,1\}^{\{a,b,c,d,e,f,g,h\}}\to \{0,1\}^{\{A,B,C,D,E,F,G,H\}}$. It is convenient to gather output cells and input cells in pairs: if~$\Sigma=\{00,01,10,11\}$, then the function can be seen as~$\phi:\Sigma^{\{ab,cd,ef,gh\}}\to\Sigma^{\{AB,CD,EF,GH\}}$. Its visibility diagram is shown in Figure \ref{fig_diag_K8}.
\begin{figure}
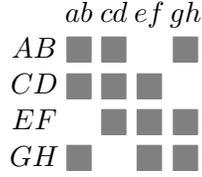

\begingroup
\arraycolsep=1pt
\begin{equation*}
\begin{array}{ccccc}
&ab&cd&ef&gh\\
AB&\slot&\slot&&\slot\\
CD&\slot&\slot&\slot&\\
EF&&\slot&\slot&\slot\\
GH&\slot&&\slot&\slot\\
\end{array}
\end{equation*}
\endgroup
\caption{The visibility diagram of a function generating~$\Mon_8$}\label{fig_diag_K8}
\end{figure}

We define a correction rule~$\rho:\Sigma^3\to \Sigma$ and then let
\begin{align*}
\mathtt{AB(gh,ab,cd)}&\mathtt{=\rho(\overline{gh},ab,cd)},\\
\mathtt{CD(ab,cd,ef)}&\mathtt{=\rho(ab,cd,ef)},\\
\mathtt{EF(cd,ef,gh)}&\mathtt{=\rho(cd,ef,gh)},\\
\mathtt{GH(ef,gh,ab)}&\mathtt{=\rho(ef,gh,\overline{ab})}.
\end{align*}

In order to define~$\rho$, we say that~$x\in \Sigma$ is a constant block if~$x=00$ or~$x=11$. If~$x\in \Sigma$, then~$x_0$ denotes its first bit. For~$x,y,z\in \Sigma$,~$\rho(x,y,z)$ is defined depending on the number~$k$ of constant blocks among~$x,y,z$:
%
\begin{itemize}
\item If~$k=0$ or~$3$, then~$\rho(x,y,z)=pp$ with~$p=\maj(x_0,y_0,z_0)$,
\item If~$k=1$, then~$\rho(x,y,z)$ is the unique constant block,
\item If~$k=2$: if~$y$ is constant, then~$\rho(x,y,z)=y$, otherwise~$\rho(x,y,z)=x_0z_0$.
\end{itemize}

We first check that if~$xyz$ is monotonic, then~$\rho(x,y,z)=y$. Observe that~$k=2$ or~$3$. If~$k=3$, then~$\rho(x,y,z)=y$. If~$k=2$ and~$y$ is constant, then~$\rho(x,y,z)=y$. If~$k=2$ and~$y$ is non-constant, then~$x\neq z$ so~$\rho(x,y,z)=x_0z_0=y$.

Therefore, if~$xyzt$ is monotonic, then~$\phi(x,y,z,t)=xyzt$, so the image of~$f$ contains~$\Mon_8$.

We show that~$\im{f}\subseteq \Mon_8$. We start with a first observation.
\begin{claim}
The function~$\phi_0:\{0,1\}^4\to \{0,1\}^4$ defined by
\[
\phi_0(x,y,z,t)=(\maj(\overline{t},x,y),\maj(x,y,z),\maj(y,z,t),\maj(z,t,\overline{x}))
\]
generates~$\Mon_4$.
\end{claim}
If~$xyzt$ is monotonic, then~$\phi_0$ sends this sequence to itself, so~$\Mon_4\subseteq\im{\phi_0}$. We show that every~$s=\phi_0(x,y,z,t)$ belongs to~$\Mon_4$. As~$\phi_0$ commutes with the action of the group~$\Di_8$, it is sufficient to show that if~$s_0=s_3$, then~$s$ is constant. Indeed, every string~$s$ can be sent to a string~$s'$ satisfying~$s'_0=s'_3$ by iterating the action of~$r\in \Di_8$: if~$s_0\neq s_3$ then there exists~$i\in [0,2]$ such that~$s_i\neq s_{i+1}$, so~$s'=r^{3-i}\cdot s$ satisfies~$s'_0=\overline{s_{i+1}}=s_i=s'_3$.

Assume that~$s_0=s_3$. One must have~$x=t$ and~$y=s_0=s_3=z$, so~$s_1=s_2=y=z$, so~$s$ is constant. The claim is proved.

Let~$k$ be the number of constant blocks among~$x,y,z,t$:
\begin{itemize}
\item If~$k=0$ or~$k=4$, then~$\phi(x,y,z,t)=ppqqrrss$ where~$pqrs=\phi_0(x_0,y_0,z_0,t_0)$. As~$pqrs$ is monotonic, so is~$ppqqrrss$.
\item If~$k=1$, we can assume by symmetry that~$y$ is constant. One has~$\phi(x,y,z,t)=yyypp$ for some~$p\in\{0,1\}$, so it is monotonic,
\item If~$k=2$ then, up so symmetry,~$x$ is constant and~$y$ or~$z$ is constant. In the first case,~$\phi(x,y,z,t)=xyy\overline{x}$ is monotonic, in the second case~$\phi(x,y,z,t)=xx_0z_0zz_0\overline{x_0}$ is monotonic,
\item If~$k=3$, we can assume by symmetry that~$t$ is non-constant. One has~$\phi(x,y,z,t)=xppzz_0\overline{x_0}$ where~$p=\maj(x_0,y_0,z_0)$, which is monotonic.\qedhere
\end{itemize}
%
\end{proof}

\begin{remark}
The definition of the function generating~$\Mon_8$ is quite technical and not as elegantly expressed as the functions generating~$\Mon_5$ and~$\Mon_7$. It would be interesting to find a clearer definition, or a clearer function.
\end{remark}

\section{Structural results}\label{sec_general}
We now prove general results about the generation of~$\Mon_n$ for any~$n$.

We start with a result for general languages, which will be applied to~$\Mon_n$, to show that the minimal complexes generating~$\Mon_n$ only have intervals as maximal simplices.

\subsection{Decomposition}
We recall from \cite{H25} that a language~$L\subseteq A^I$ is not irreducible if there exists a non-trivial partition~$I=X\sqcup Y$ such that for every~$w\in A^I$, if~$\restr{w}{X}$ and~$\restr{w}{Y}$ have extensions in~$L$, then~$w\in L$. Intuitively, it means that the contents of~$X$ and~$Y$ can be chosen independently. We showed in \cite{H25} that in this case, every minimal complex generating~$L$ is made of simplices contained in either~$X$ or~$Y$, hence allowing no communication between~$X$ and~$Y$.

For~$n\geq 3$, the language~$\Mon_n$ is irreducible so this result cannot be applied. However, we prove a generalization of this result by relaxing the condition that~$X$ and~$Y$ are disjoint.
\begin{definition}\label{def_decomposition}
Let~$X,Y\subseteq I$ satisfy~$X\cup Y=I$. We say that~$X$ and~$Y$ \textbf{decompose}~$L\subseteq A^I$ if for every~$w\in A^I$, if~$\restr{w}{X}$ and~$\restr{w}{Y}$ have extensions in~$L$, then~$w\in L$.
\end{definition}

Intuitively, when fixing the content of~$X\cap Y$, the contents of~$I\setminus X$ and~$I\setminus Y$ can be filled independently of each other.

\begin{theorem}\label{thm_decomposition}
Let~$X\cup Y=I$ and assume that~$X$ and~$Y$ decompose~$L$. If~$K$ is a minimal complex generating~$L$, then every maximal simplex intersecting both~$X$ and~$Y$ intersects~$X\cap Y$.
\end{theorem}

\begin{proof}
Let~$J$ be the set of maximal simplices of~$K$,~$B$ a finite set and~$f:B^J\to A^I$ generate~$L$, such that~$K_f\subseteq K$ (thus~$K_f=K$ by minimality of~$K$). This function exists by Proposition \ref{prop_canonical}. Assume for a contradiction that for some~$p\in J$,~$\W^f(p)$ intersects both~$X$ and~$Y$ but not~$X\cap Y$, i.e.~$\W^f(p)\subseteq X\cup Y$.

The idea is to replace the input cell~$p$ by two new input cells~$p_X$ and~$p_Y$. Every output cell~$i\in I$ that was reading~$p$ now reads either~$p_X$ or~$p_Y$ depending on whether~$i\in X$ or~$i\in Y$. The problem is that in general~$p_X$ and~$p_Y$ will not store the same value, so for a general language the output has no reason to be correct, but the decomposition assumption implies that the output is indeed correct.

More precisely, let~$J'=J\setminus \{p\}\cup\{p_X,p_Y\}$. Let~$\varphi:J'\to J$ send~$p_X$ and~$p_Y$ to~$p$ and every other element to itself. Let~$\psi_X,\psi_Y:J\to J'$ send~$p$ to~$p_X$ and~$p_Y$ respectively, and every other element to itself. Note that~$\varphi\circ \psi_X=\varphi\circ \psi_Y=\id_J$. For~$u\in B^{J'}$, let~$u_X=x\circ \psi_X$ (resp.~$u_Y=u\circ \psi_Y$) be obtained by copying the content of~$p_X$ (resp.~$p_Y$) in~$p$.

We define~$g:B^{J'}\to A^I$ as follows: for~$u\in B^{J'}$ and~$i\in I$,
\[
g_i(u)=\begin{cases}
f_i(u_X)&\text{if }i\in X,\\
f_i(u_Y)&\text{if }i\in Y.
\end{cases}
\]
We first show that it is well-defined when~$i\in X\cap Y$. The sequences~$u_X$ and~$u_Y$ coincide everywhere except possible at~$p$. Therefore, when~$i\in X\cap Y$, one has~$i\notin \W^f(p)$ so~$f_i(u_X)=f_i(u_Y)$.

One easily has~$L=\im{f}\subseteq \im{g}$. Indeed, for each~$u\in B^J$ one has~$g(u\circ \varphi)=f(u)$ because for each~$i\in I$,~$g_i(u\circ \varphi)$ equals~$f_i(u\circ \varphi\circ\psi_X)=f_i(u)$ or~$f_i(u\circ \varphi\circ\psi_Y)=f_i(u)$.

We finally show that~$\im{g}\subseteq L$. Let~$u\in B^{J'}$ and~$w=g(u)$. By definition of~$g$,~$w$ coincides with~$f(u_X)$ on~$X$ and with~$f(u_Y)$ on~$Y$, so the restrictions of~$w$ to~$X$ and~$Y$ both have extensions in~$L$. It implies that~$w\in L$ as~$X$ and~$Y$ decompose~$L$.

The communication complex~$K_g$ is properly contained in~$K$. Indeed,
\begin{align*}
\W^g(p_X)&=\W^f(p)\cap X\subseteq \W^f(p)\setminus Y\subsetneq \W^f(p),\\
\W^g(p_Y)&=\W^f(p)\cap Y\subseteq \W^f(p)\setminus X\subsetneq\W^f(p),\\
\W^g(q)&=\W^f(q) \text{ for every~$q\in J\setminus \{p\}$.}
\end{align*}
It contradicts the minimality of~$K$.
\end{proof}
As a corollary, we obtain the previously mentioned result about non-irreducible languages: if~$X$ and~$Y$ are disjoint, then no maximal simplex of~$K$ intersects both~$X$ and~$Y$, because it cannot intersect~$X\cap Y=\emptyset$.

\subsection{Communication only occurs via intervals}\label{sec_intervals}
We now apply this result to~$\Mon_n$.
\begin{theorem}[Minimal complexes are made of intervals]\label{thm_intervals}
Let~$K$ be a minimal complex generating~$\Mon_n$. The maximal simplices of~$K$ are intervals.
\end{theorem}
\begin{proof}
Let~$0\leq a<b<n$ be distinct. We show that the intervals~$\int{a,b}$ and~$\int{b,a}$ decompose~$\Mon_n$. The reason is that if~$w$ is monotonic, then the values of~$w_a$ and~$w_b$ uniquely determine the values of~$w$ on either~$\int{a,b}$ or~$\int{b,a}$.

Assume that the restrictions of~$w\in\{0,1\}^n$ to~$\int{a,b}$ and~$\int{b,a}$ are monotonic. There are two cases:
\begin{itemize}
\item If~$w_a=w_b=v$, as~$\restr{w}{\int{a,b}}$ is monotonic, it takes constant value~$v$. As~$\restr{w}{\int{b,a}}$ is monotonic, so is~$w$.
\item If~$w_a=\overline{w_b}=v$, as~$\restr{w}{\int{b,a}}$ is monotonic,~$w$ takes constant value~$v$ on~$\int{0,a}$ and constant value~$\overline{v}$ on~$\int{b,n-1}$. As~$\restr{w}{\int{a,b}}$ is monotonic, so is~$w$.
\end{itemize}
In any case,~$w$ is monotonic, so~$\int{a,b}$ and~$\int{b,a}$ indeed decompose~$\Mon_n$.


Let~$S\subseteq I_n$ be a maximal simplex of~$K$. If~$S$ is not an interval, then there exist distinct elements~$a,b\in I_n\setminus S$ such that~$S$ intersects both~$\int{a,b}$ and~$\int{b,a}$. It contradicts Theorem \ref{thm_decomposition}, as~$S$ does not intersect~$\int{a,b}\cap\int{b,a}=\{a,b\}$.
\end{proof}

\subsection{Vertex deletion and insertion}\label{sec_insertion}
There are simple ways to transform complexes generating~$\Mon_n$ into complexes generating~$\Mon_{n+1}$ and vice versa.

\subsubsection{Vertex deletion}
The first observation is that the operation of removing a bit sends the set of monotonic sequences of length~$n+1$ to the set of monotonic sequences of length~$n$. Therefore, any complex that generates~$\Mon_{n+1}$ can be easily transformed into a complex generating~$\Mon_n$ by simply removing a vertex, and renaming the other vertices. Let us formalize this idea.

Let~$i\in [0,n]$. The set~$[0,n]\setminus \{i\}$ can be identified with~$[0,n-1]$ using the functions~$\rho^-_i:[0,n]\setminus \{i\}\to [0,n-1]$ and~$\rho^+_i:[0,n-1]\to [0,n]\setminus \{i\}$ defined by
\[
\rho^-_i(j)=\begin{cases}
j&\text{if }j<i,\\
j-1&\text{if }j>i,
\end{cases}
\quad\text{and}\quad
\rho^+_i(j)=\begin{cases}
j&\text{if }j<i,\\
j+1&\text{if }j\geq i.
\end{cases}
\]
These functions are inverses of each other.

\begin{definition}[Vertex deletion]
If~$K$ is a complex over~$[0,n]$ then the complex obtained by deletion of~$i$ in~$K$ is the complex over~$[0,n-1]$ defined by
\begin{align*}
\del{i}{S}&=\rho^-_i(S\setminus\{i\}),\\
\del{i}{K}&=\{\del{i}{S}:S\in K\}.
\end{align*}
\end{definition}

It is indeed a complex: if~$T\subseteq \del{i}{S}$ for some~$S\in K$, then~$T=\del{i}{T'}$ where~$T'=\rho_i^+(T)$ belongs to~$K$ because~$T'$ is contained in~$S$.

In the graphical representation of the complex~$K$, the new complex~$\del{i}{K}$ is obtained by removing row~$i$, shifting the lower rows upwards and removing the non-maximal intervals (see Figure \ref{fig_deletion}).
\begin{figure}[ht]
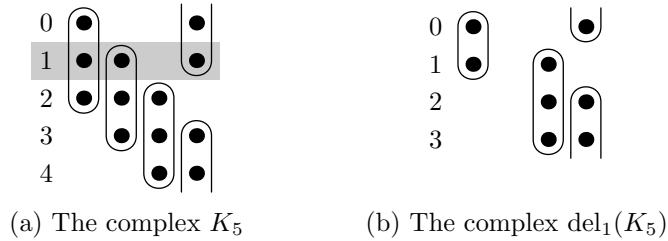

\centering
\subcaptionbox{The complex~$K_5$}[.4\textwidth]{\includegraphics{Images/complex-12}}
\subcaptionbox{The complex~$\del{1}{K_5}$}[.4\textwidth]{\includegraphics{Images/complex-13}}
\caption{Vertex deletion}\label{fig_deletion}
\end{figure}

\begin{proposition}[Vertex deletion and generation]\label{prop_deletion}
Let~$i\in [0,n]$. If a complex~$K$ over~$[0,n]$ generates~$\Mon_{n+1}$, then~$\del{i}{K}$ generates~$\Mon_n$.
\end{proposition}
\begin{proof}
For simplicity of notation, we can assume by symmetry that~$i=n$, so~$\rho^-_i$ is the identity. The function~$f:\{0,1\}^{[0,n]}\to \{0,1\}^{[0,n-1]}$ that removes the last bit satisfies~$f(\Mon_{n+1})=\Mon_n$, and
\[
\W^f(j)=\begin{cases}
\emptyset&\text{if }j=n\\
\{j\}&\text{otherwise.}
\end{cases}
\]
Therefore,~$f_*(K)$ generates~$\Mon_n$ by Proposition \ref{prop_image}. If~$S\in K$, then~$f_*(S)=\bigcup_{j\in S}\W^f(j)=S\setminus \{n\}$. Finally,~$f_*(K)$ is the complex induced by these sets, and we have seen that these sets already form a complex.
\end{proof}

However, if~$K$ is minimal generating~$\Mon_{n+1}$, then~$\del{i}{K}$ might not be minimal generating~$\Mon_n$. For instance, we will see in Section \ref{sec_K5_family} that~$K_5$ is minimal generating~$\Mon_5$, however~$\del{1}{K_5}$ is not minimal as it properly contains~$\comp{I_4\setminus \{0\},I_4\setminus\{1\}}$ which already generates~$\Mon_4$ (see Figure \ref{fig_deletion}).


\subsubsection{Vertex insertion}
Conversely, there is a simple way to extend a complex generating~$\Mon_n$ into a complex generating~$\Mon_{n+1}$, based on the following observation: every~$w\in\Mon_{n+1}$ can be obtained by starting from some~$u\in\Mon_n$ and inserting a bit at some position~$i$. The possible value(s) of the inserted bit only depend on the values of~$u_{i-1}$ and~$u_{i}$. Therefore, the new cell only needs to communicate with~$i-1$ and~$i$, by reading all the input cells read by them.

Again, the corresponding transformation of the communication complex can be easily expressed: it consists in adding a new vertex, adding it to all the maximal simplices containing~$i-1$ or~$i$, and renaming the vertices. 

\begin{definition}[Vertex insertion]\label{def_insertion}
Let~$i\in [0,n]$. If~$K$ is a complex over~$[0,n-1]$, then~$\ins{i}{K}$ is the complex over~$[0,n]$ induced by the sets
\[
\ins{i}{S}=
\begin{cases}
\rho_i^{+}(S)\cup \{i\}&\text{if $S$ contains~$i-1\bmod n$ or~$i\bmod n$},\\
\rho_i^{+}(S)&\text{otherwise,}
\end{cases}
\]
where~$S\in K$.
\end{definition}

The graphical representation of~$\ins{i}{K}$ is obtained by inserting a row at position~$i$, and copying the contents of rows with former indices~$i-1\bmod n$ and~$i\bmod n$ into it (see Figure \ref{fig_insertion}).
\begin{figure}[ht]
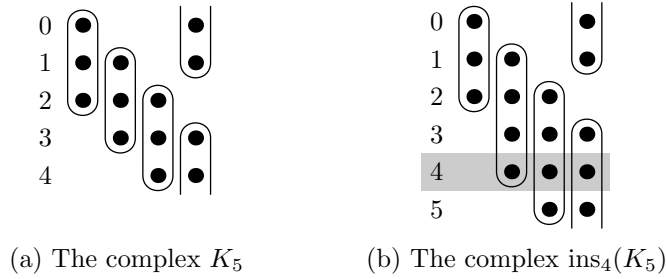

\centering
\subcaptionbox{The complex~$K_5$}[.4\textwidth]{\includegraphics{Images/complex-14}}
\subcaptionbox{The complex~$\ins{4}{K_5}$}[.4\textwidth]{\includegraphics{Images/complex-15}}
\caption{Vertex insertion}\label{fig_insertion}
\end{figure}

%
%
%

\begin{proposition}[Vertex insertion and generation]\label{prop_insertion}
Let~$i\in[0,n]$. A complex~$K$ generates~$\Mon_n$ if and only if~$\ins{i}{K}$ generates~$\Mon_{n+1}$.
\end{proposition}
\begin{proof}
First, if~$\ins{i}{K}$ generates~$\Mon_{n+1}$, then~$K=\del{i}{\ins{i}{K}}$ generates~$\Mon_n$ by Proposition \ref{prop_deletion}.

Now assume that~$K$ generates~$\Mon_n$. For simplicity of notation, we assume by symmetry that~$i=n$, so the new vertex is~$n$ and there is no need to rename the other vertices, i.e.~$\rho_i^+$ is the identity.

The idea is very simple: every output cell in~$[0,n-1]$ uses the same generation procedure generating~$\Mon_n$, and the new cell~$n$ reads all the input cells read by~$0$ and by~$n-1$, evaluates their output values and then outputs any value that makes the extended sequence monotonic (if~$0$ and~$n-1$ take the same value, then~$n$ can take any value, otherwise~$n$ takes the same value as~$n-1$). The new output cell has moreover access to a new input cell~$p$ containing an arbitrary number~$k\in\{0,1\}$ so that if the two values~$0$ and~$1$ are possible, then it outputs~$k$ (if only one output value~$a$ is possible, then it ignores~$k$ and outputs~$a$).

Let~$f:B^J\to\Mon_n$ generate~$\Mon_n$ with~$K_f\subseteq K$. Let~$J'=J\cup\{p\}$ where~$p\notin J$. The new generation procedure for~$\Mon_{n+1}$ is a function~$g:B^{J'}\to \Mon_{n+1}$. One has~$\W^g(p)=\{n\}$ and for~$j\in J$,
\[
\W^g(j)=\begin{cases}
\W^f(j)&\text{if }\W^f(j)\cap\{0,n-1\}=\emptyset,\\
\W^f(j)\cup\{n\}&\text{otherwise}.
\end{cases}
\]
Therefore,~$K_g=\ins{n}{K_f}\subseteq\ins{n}{K}$, so~$\ins{n}{K}$ generates~$\Mon_{n+1}$.
\end{proof}

\begin{example}\label{ex_K2}
For~$n=2$, the complex~$K_2=\comp{\{0\},\{1\}}$ generates~$\Mon_2=\{0,1\}^2$, via the identity~$\id:\{0,1\}^2\to\{0,1\}^2$. The complexes that can be obtained from~$K_2$ by successively inserting vertices are the complexes of the form~$\{I_n\setminus\{i\},I_n\setminus\{j\}\}$, which indeed generate~$\Mon_n$ as mentioned at the beginning of Section \ref{sec_mon_examples}.
\end{example}

If~$K$ is minimal generating~$\Mon_n$, then the complex~$\ins{i}{K}$ might not be minimal generating~$\Mon_{n+1}$, i.e.~some maximal intervals in~$\ins{i}{K}$ may be replaced by smaller intervals. We show that some of the maximal intervals of~$\ins{i}{K}$ cannot be removed.
%
%
%
%
%
%
\begin{proposition}\label{prop_insertion_minimal}
Let~$n\geq 2$ and~$K$ be minimal generating~$\Mon_n$. Let~$i\in [0,n]$ and~$K'\subseteq \ins{i}{K}$ generate~$\Mon_{n+1}$. If~$I$ is a maximal interval of~$K$ containing~$\{i-1\bmod n,i\bmod n\}$ or disjoint from~$\{i-1\bmod n,i\bmod n\}$, then~$K'$ contains~$\ins{i}{I}$.
\end{proposition}
\begin{proof}
For simplicity of notation, we can assume by symmetry that~$i=n$.

Observe that~$\del{n}{\ins{n}{K}}=K$. If~$K'\subseteq\ins{n}{K}$ generates~$\Mon_{n+1}$, then~$\del{n}{K'}\subseteq \del{n}{\ins{n}{K}}=K$ generates~$\Mon_n$, so~$\del{n}{K'}=K$ by minimality of~$K$.

Therefore, if~$I$ is a maximal interval of~$K$, then there exists an interval~$J$ of~$K'$ such that~$I\subseteq \del{n}{J}$. If~$I$ satisfies the conditions in the statement, then we show that~$\ins{n}{I}\subseteq J$, implying that~$\ins{n}{I}\in K'$. As~$I=\del{n}{J}$, we already know that~$J$ contains~$I$.
\begin{itemize}
\item If~$I$ contains neither~$0$ nor~$n-1$, then~$\ins{n}{I}=I\subseteq J$.
\item If~$I$ contains both~$0$ and~$n-1$, then so does~$J$. As~$J$ is an interval,~$J$ contains~$n$ or~$[0,n-1]$. The latter case is impossible because it implies that~$I=\del{n}{J}=[0,n-1]$ so~$[0,n-1]\in K$, contradicting the minimality of~$K$ (and~$n\geq 2$). Therefore,~$n\in J$, hence~$\ins{n}{I}\subseteq J$.\qedhere
\end{itemize}
%
%
%
%
\end{proof}

\begin{corollary}\label{cor_insertion_minimal}
Let~$n\geq 3$,~$K$ be minimal generating~$\Mon_n$ and~$i\in [0,n]$. Assume that every maximal interval of~$K$ contains or is disjoint from~$\{i-1\bmod n,i\bmod n\}$. The complex~$\ins{i}{K}$ is minimal generating~$\Mon_{n+1}$.
\end{corollary}
This condition can be expressed as~$i-1\bmod n$ and~$i\bmod n$ belonging to the same simplices, or reading the same input cells.

We now weaken the condition in Corollary \ref{cor_insertion_minimal}.
\begin{proposition}[Vertex insertion and minimality]\label{prop_insertion_minimal2}
Let~$n\geq 2$,~$K$ be minimal generating~$\Mon_n$ and~$i\in[0,n]$. Assume that every interval of~$K$ containing~$i\bmod n$ contains~$i-1\bmod n$ (or symmetrically, every interval containing~$i-1\bmod n$ contains~$i\bmod n$). The complex~$\ins{i}{K}$ is minimal generating~$\Mon_{n+1}$.
\end{proposition}

\begin{proof}
Again, we assume by symmetry that~$i=n$ for simplicity.

First observe that~$\ins{n}{K}\subseteq\ins{n-1}{K}$. Indeed, let~$I$ be an interval of~$K$:
\begin{itemize}
\item If~$n-1\in I$, then~$\ins{n}{I}=\ins{n-1}{I}$,
\item  If~$n-1\notin I$, then~$0\notin I$ so~$\ins{n}{I}=I\subseteq \ins{n-1}{I}$.
\end{itemize}

Now, let~$K'\subseteq \ins{n}{K}$ generate~$\Mon_{n+1}$ and let~$I$ be a maximal interval of~$K$. We show that~$K'$ contains~$\ins{n}{I}$. If~$I$ contains or is disjoint from~$\{0,n-1\}$, then~$K'$ contains~$\ins{n}{I}$ by Proposition \ref{prop_insertion_minimal}.

If~$I$ contains~$n-1$ but not~$0$, then~$I$ contains~$n-2$ ($I$ is maximal so it is not a singleton as~$n\geq 3$), so~$\ins{n}{I}=\ins{n-1}{I}$. We apply Proposition \ref{prop_insertion_minimal} to~$i=n-1$. As~$I$ contains both~$n-2$ and~$n-1$, and~$K'\subseteq \ins{n}{K}\subseteq \ins{n-1}{K}$,~$K'$ contains~$\ins{n-1}{I}$ by Proposition \ref{prop_insertion_minimal}.

We have covered all the possible maximal intervals~$I$, so~$K'=\ins{n}{K}$, which is minimal.
%
\end{proof}

The vertex insertion illustrated in Figure \ref{fig_insertion} satisfies the conditions of Proposition \ref{prop_insertion_minimal2}: every maximal simplex of~$K_5$ containing~$4$ also contains~$3$, so if~$K_5$ is minimal generating~$\Mon_5$ (which we will establish in Section \ref{sec_K5_family}), then~$\ins{4}{K_5}$ is minimal generating~$\Mon_6$.

\section{A framework for discarding a complex}\label{sec_rule_system}
The main way to prove that a complex does not generate a language is to assume that it does and to successively derive constraints expressing that certain partial inputs are sent to certain partial outputs, and eventually derive a contradiction from conflicting constraints that send compatible partial inputs to incompatible partial outputs. Most of the proofs of such results in \cite{H25} and in this article follow this pattern and we find it clarifying to present a unifying framework in which these arguments can be expressed. We will however not push the formalization too far in order to make the arguments more human-friendly, and we will not explicitly refer to this framework in the proofs, but only use it as a guideline.



This framework is essentially a logical inference system, whose interest is twofold:
it identifies the common structure in many arguments, hopefully clarifying them, and it can be turned into an automatic tool for finding proofs. Some of the arguments of this article have been found by a computer program exhaustively searching for conflicts (the program is available at \cite{Prover}). However, such an automatic approach quickly reaches its limits due to the combinatorial explosion of the search space. Moreover, it does not directly solve the problem of finding the whole family of complexes generating a language, and is only used as a tool that helps discarding particular complexes and testing hypotheses.
 

\subsection{The inference system}
We now present this framework.

Let~$A$ and~$I$ be finite sets. A \textbf{partial sequence} of~$A^I$ is an element~$\alpha$ of~$A^{D}$ for some~$D\subseteq I$ called the \textbf{domain} of~$\alpha$ and written~$\dom(\alpha)$. Two partial sequences~$\alpha,\beta$ are \textbf{compatible} if they have a common extension, i.e.~if they coincide on~$\dom(\alpha)\cap\dom(\beta)$. Their least common extension is denoted by~$\alpha\cup\beta$.

Let~$f:B^J\to A^I$ be a function. A \textbf{partial intput} is a partial sequence of~$B^J$, a \textbf{partial output} is a partial sequence of~$A^I$. A \textbf{constraint} is a statement~$f(\alpha)\succeq p$ where~$\alpha$ is a partial input and~$p$ is a partial output. It means that for every input~$x\in B^J$ extending~$\alpha$,~$f(x)$ extends~$p$. We will sometimes write a constraint~$f(\alpha)\succeq p$ as~$f_D(\alpha)=p$, where~$D=\dom(p)$. Moreover, if~$\dom(p)$ is a singleton~$\{i\}$ and~$p$ assigns value~$v$ to~$i$, we will write~$f(\alpha)\succeq p$ as~$f_i(\alpha)=v$.

We now restrict our attention to binary languages, i.e.~subsets of~$\{0,1\}^I$ for some~$I$. A \textbf{rule} is a pair~$(p,q)$ of partial binary sequences, written~$p\to q$. A partial binary sequence~$r$ \textbf{respects} a rule~$p\to q$ if~$r$ satisfies the implication
\[
\text{$r$ extends $p$} \implies \text{$r$ extends $q$,}
\]
and a language~$L$ respects a rule if every element of~$L$ respects that rule.

\begin{proposition}
For a finite set~$I$, every binary language~$L\subseteq \{0,1\}^I$ is uniquely determined by the set of rules it respects. Moreover, a partial sequence has an extension in~$L$ if and only if it respects every rule respected by~$L$.
\end{proposition}
\begin{proof}
Let~$r$ be a partial sequence in~$\{0,1\}^I$ that has no extension in~$L$. There exists a minimal set~$E\subseteq \dom(r)$ such that~$\restr{r}{E}$ has no extension in~$L$. As~$L$ is non-empty,~$E$ is non-empty as well. Let~$i\in E$ and~$p=\restr{r}{E\setminus \{i\}}$. By minimality of~$E$,~$p$ has an extension~$y\in L$, let~$q=\restr{y}{E}$. Observe that for any~$z\in L$ extending~$p$, one must have~$z_i\neq r_i$, otherwise~$\restr{r}{E}=\restr{z}{E}$ has an extension~$z$ in~$L$. In particular, one has~$y_i\neq r_i$, so~$r$ does not extend~$q$ hence~$r$ does not respect the rule~$p\to q$.

However, every~$z\in L$ respects~$p\to q$. Indeed, if~$z$ extends~$p$ then we have shown that~$z_i\neq r_i$, implying~$z_i=y_i$ (we are using the assumption that the alphabet is binary), hence~$z$ extends~$q$. Therefore, the rule~$p\to q$ is respected by~$L$ but not by~$r$.
\end{proof}

Note that languages over larger alphabets cannot in general be fully described by such rules.

We recall from Proposition \ref{prop_canonical} that if~$K$ generates~$L$, then there is a generating function~$f:L^J\to L$ satisfying~$K_f\subseteq K$, where~$J$ is the set of maximal simplices of~$K$. Most of the arguments showing that a complex~$K$ does not generate a given binary language~$L$ consist in assuming the existence of such a generating function and deriving constraints using the following rules:
\begin{description}
\item[Language axioms.] For each~$w\in L$, there exists~$x\in L^J$ such that~$f(x)=w$ (and we can assume that~$x=(w,\ldots,w)$),
\item[Language rules.] Let~$p\to q$ be a rule respected by~$L$. If~$f(\alpha)\succeq p$ then~$f(\alpha)\succeq q$,
\item[Restriction rule.] If~$f_i(\alpha)=a$, then~$f_i(\restr{\alpha}{\W_f(i)})=a$,
\item[Join rule.] If~$\alpha,\beta$ are compatible,~$f(\alpha)\succeq p$ and~$f(\beta)\succeq q$ and~$p,q$ are compatible, then~$f(\alpha\cup\beta)\succeq p\cup q$.
\end{description}
The argument then ends with:
\begin{description}
\item[Conflict rule.] If~$\alpha,\beta$ are compatible,~$f(\alpha)\succeq p$ and~$f(\beta)\succeq q$ and~$p,q$ are incompatible, then we obtain a contradiction, showing that~$f$ cannot exist.
\end{description}
Typically, the conflict rule will be applied to compatible partial inputs~$\alpha,\beta$ such that~$f_i(\alpha)=0$ and~$f_i(\beta)=1$ for some~$i\in I$.

The language axioms express the surjectivity of~$f$ (the specific form of the inputs comes from Proposition \ref{prop_canonical}). The language rules capture the correlations imposed by the language. The restriction and join rules exploit the lack of certain simplices in the complex~$K_f$ by reducing the domains of the involved partial inputs and eventually producing contradictory statements.

It is not difficult to see that this inference system is sound, in the sense that any constraint~$f(\alpha)\succeq p$ that can be derived using these rules is indeed true under the assumption that~$f$ exists, and that if the conflict rule can be applied then it gives a contradiction and~$f$ cannot exist.

However, this system is incomplete: in \cite{H25} we proved that the complex~$K=\comp{\int{1,4},\int{0,2},\int{4,1},\{0,2,3\},\{0,2,4\}}$ does not generate~$\mathsf{U}_5$, which is the set of binary strings of length~$5$ containing one occurrence of~$1$. An exhaustive search by a computer program shows that the inference system does not find any conflict for~$K$ and~$\mathsf{U}_5$.

\subsection{The rules of \texorpdfstring{$\Mon_n$}{Mon}}

A partial sequence~$p$ has an extension in~$\Mon_n$ if and only if it respects the following rules:
%
 for~$0\leq i<j<k<n$ and~$v\in\{0,1\}$:
\begin{itemize}
\item If~$p_i=p_k=v$, then~$p_j=v$,
\item If~$\overline{p_i}=p_j=v$, then~$p_k=v$,
\item If~$p_j=\overline{p_k}=v$, then~$p_i=v$.
\end{itemize}

\subsection{Key technical tool}
In this section we prove a technical result that will extensively be used in the analysis of~$\Mon_n$, in order to prove the minimality of the previous complexes as well as more general results. The proof implicitly uses the inference system.

Monotonic sequences are rigid in the sense that the values of two cells sometimes determine the values of all the other cells. For instance, if~$y\in\Mon_n$ satisfies~$y_0=y_{n-1}=0$, then~$y_j=0$ for all~$j$. Therefore, in a generation procedure, cells need to communicate in order to account for these correlations. The next result captures this idea, implying in particular that if a complex~$K$ generates~$\Mon_n$, then it contains many intervals: for all~$i,j\in I_n$, at least one of the intervals~$\int{i+1,j}$ and~$\int{j,i}$ must belong to~$K$.

Given a function~$f:B^J\to A^I$, we recall that if~$S\subseteq I$, then~$\up[f]{S}=\bigcap_{i\in S}\W_f(i)$ is the set of input cells visible by all the output cells in~$S$, and is non-empty if and only if~$S\in K_f$.
\begin{theorem}[Main tool]\label{thm_main_tool}
Let~$i,j\in I_n$ and~$v\in\{0,1\}$. If~$f$ generates~$\Mon_n$, then there exists a partial input~$\alpha$ such that
\begin{itemize}
\item $\dom(\alpha)=\up[f]{\int{i+1,j}}\cup\up[f]{\int{j,i}}$,
\item $f_j(\alpha)=v$.
\end{itemize}
\end{theorem}
This result implies in particular that at least one of the intervals~$\int{i+1,j}$ and~$\int{j,i}$ must belong to~$K$, otherwise~$\dom(\alpha)$ would be empty hence~$f_j$ would constantly equal~$v$, which is a contradiction as it takes both values~$0$ and~$1$.

 By symmetry, it is sufficient to prove the result for~$i=n-1$ and~$v=0$: for any~$j$, there exists~$\alpha$ such that~$\dom(\alpha)=\up[f]{\int{0,j}}\cup\up[f]{\int{j,n-1}}$ and~$f_j(\alpha)=0$. 
It will follow from the following more general statement. In order to lighten the notations, we now drop the subscript~$f$ in~$\W_f$.

\begin{lemma}\label{lem_mono}
Let~$0\leq i\leq j \leq k<n$ and~$v\in\{0,1\}$. Let~$\alpha,\beta$ be compatible partial inputs such that~$f_i(\alpha)=f_k(\beta)=v$. Then~$f_j(\gamma)=v$, where
\begin{equation*}
\gamma=\alpha\res{\up{\int{i,j}}}\cup\beta\res{\up{\int{j,k}}}.
\end{equation*}
\end{lemma}

\begin{proof}
We assume that~$v=0$, the other case is similar. 
We first show a weaker statement.

\begin{claim}\label{claim_left}
Under the same assumptions, one has~$f_j(\mu)=0$, where
\begin{equation*}
\mu=\alpha\res{\up{\int{i,j}}}\cup\beta\res{\W(j)}.
\end{equation*}
\end{claim}
We fix~$0\leq i\leq k<n$ and prove it by induction on~$j\in [i,k]$. For~$j=i$, one has~$\mu=\restr{\alpha}{\W(i)}\cup\restr{\beta}{\W(i)}$, which extends~$\restr{\alpha}{\W(i)}$ so indeed~$f_i(\mu)=0$. We assume the result for~$j<k$ and prove it for~$j+1$. Let~$\mu_j=\restr{\alpha}{\up{\int{i,j}}}\cup \restr{\beta}{\W(j)}$. By induction hypothesis, one has~$f_j(\mu_j)=f_k(\beta)=0$. As~$\mu_j$ and~$\beta$ are compatible,~$f_{j+1}(\restr{(\mu_j\cup\beta)}{\W(j+1)})=0$. One has
\begin{align*}
\restr{(\mu_j\cup\beta)}{\W(j+1)}&=\restr{\alpha}{\up{\int{i,j+1}}}\cup\beta\res{\W(j)\cap\W(j+1)}\cup\beta\res{\W(j+1)}\\
&=\restr{\alpha}{\up{\int{i,j+1}}}\cup\restr{\beta}{\W(j+1)}\\
&=\mu_{j+1}
\end{align*}
which proves the induction step. The claim is proved.

We then use this claim to prove the statement of the lemma by decreasing induction on~$j$.

For~$j=k$, one has~$\gamma=\alpha\res{\up{\int{i,k}}}\cup\restr{\beta}{\W(k)}$ which extends~$\restr{\beta}{\W(k)}$, so indeed~$f_k(\gamma)=0$. We assume the result for~$j>i$ and prove it for~$j-1$. Let~$\gamma_j=\alpha\res{\up{\int{i,j}}}\cup\beta\res{\up{\int{j,k}}}$. By induction hypothesis, one has~$f_i(\alpha)=f_j(\gamma_j)=0$. As~$\alpha$ and~$\gamma_j$ are compatible, we can apply Claim \ref{claim_left}, implying~$f_{j-1}(\delta)=0$, where
\begin{align*}
\delta&=\alpha\res{\up{\int{i,j-1}}}\cup \gamma_j\res{\W(j-1)}\\
&=\alpha\res{\up{\int{i,j-1}}}\cup \alpha\res{\up{\int{i,j}}}\cup\beta\res{\up{\int{j-1,k}}}\\
&=\alpha\res{\up{\int{i,j-1}}}\cup\beta\res{\up{\int{j-1,k}}}\\
&=\gamma_{j-1}
\end{align*}
which proves the induction step.
\end{proof}

\begin{proof}[Proof of Theorem \ref{thm_main_tool}]
We prove the statement for~$i=n-1$ and~$v=0$, the other cases hold by symmetry. Let~$j\in I_n$, i.e.~$0\leq j<n$. Let~$\alpha$ be an input such that~$f(\alpha)=0\ldots 0$. By Lemma \ref{lem_mono} applied to~$\beta=\alpha$,~$i'=0$,~$j$ and~$k=n-1$, one has~$f_j(\gamma)=0$ where~$\gamma=\restr{\alpha}{\up{\int{0,j}}}\cup\restr{\alpha}{\up{\int{j,n-1}}}$, so~$\gamma$ satisfies the prescribed conditions.
\end{proof}

\begin{remark}[Symmetric versions]\label{rmk_lem_sym}
By symmetry, Lemma \ref{lem_mono} implies more generally the following statements. Let~$v\in \{0,1\}$ and~$\alpha,\beta$ be compatible partial inputs:
\begin{itemize}
\item If~$0\leq i<j\leq k<n$,~$f_i(\beta)=\overline{v}$ and~$f_j(\alpha)=v$, then~$f_k(\gamma)=v$, where
\[
\gamma=\restr{\alpha}{\up{\int{j,k}}}\cup\restr{\beta}{\up{\int{k,i}}},
\]
\item If~$0\leq i\leq j< k<n$,~$f_j(\beta)=v$ and~$f_k(\alpha)=\overline{v}$, then~$f_i(\delta)=v$, where
\[
\delta=\restr{\alpha}{\up{\int{k,i}}}\cup \restr{\beta}{\up{\int{i,j}}}.
\]
\end{itemize}
They can be obtained from Lemma \ref{lem_mono} by making some power of~$r\in\Di_{2n}$ act on~$I_n$ and~$\Mon_n$.
\end{remark}

\subsection{An even more technical result}
The following technical result is at the core of some of the next proofs. Its statement is seemingly very \emph{ad hoc}, but was obtained by identifying common patterns in several arguments. We postpone its proof to the appendix (Section \ref{app_lem}).
 
\begin{lemma}\label{lem_technical}
Let~$1<i<j<n-1$. If a complex~$K$ generates~$\Mon_n$, then~$K$ contains at least one of the following intervals:
\[
\int{0,j},\int{j+1,i},\int{i+1,1},\int{i,0},\int{1,n-1}.
\]
\end{lemma}

Of course, this lemma comes with its symmetric versions: if~$K$ generates~$\Mon_n$, then for any symmetry~$g\in \Di_{2n}$,~$g\cdot K$ generates~$\Mon_n$ hence contains one of these intervals.

\section{Families of minimal complexes generating \texorpdfstring{$\Mon_n$}{Mon}}\label{sec_families}

We now identify families of minimal complexes generating~$\Mon_n$. Each family is obtained from a particular complex by inserting vertices at certain positions. In particular, we show that the complexes presented in Section \ref{sec_mon_examples} are all minimal.

\subsection{The \texorpdfstring{$K_2$}{K2} family}\label{sec_K2_family}
We first identify all the minimal complexes generating~$\Mon_n$ and containing several~$(n-1)$-intervals. It is the first application of Theorem \ref{thm_main_tool}.

\begin{definition}
The \textbf{$K_2$ family} is the family of complexes
\[
\comp{I_n\setminus \{a\},I_n\setminus \{b\}}
\]
for~$n\geq 2$ and distinct~$a,b\in I_n$.
\end{definition}
They are exactly the complexes that can be obtained from~$K_2$ by vertex insertions, as mentioned in Example \ref{ex_K2}. They are illustrated in Figure \ref{fig_Kab}.
\begin{figure}[ht]
\centering
\includegraphics{Images/complex-0}
\caption{The complex~$\comp{I_n\setminus \{a\},I_n\setminus \{b\}}$, for $n=5$, $a=1$ and $b=4$}\label{fig_Kab}
\end{figure}

\begin{proposition}[Several~$(n-1)$-intervals]\label{prop_mono_2_simplices}
Let~$n\geq 2$. The minimal complexes generating~$\Mon_n$ and having several~$(n-1)$-intervals are the members of the~$K_2$ family.
\end{proposition}
\begin{proof}
We show that~$K_{a,b}=\comp{I_n\setminus \{a\},I_n\setminus \{b\}}$ is minimal. Let~$K\subseteq K_{a,b}$ generate~$\Mon_n$. We use the fact that~$K$ does not contain the edge~$\{a,b\}$ to show that~$K$ must contain~$K_{a,b}$. We show that~$K$ contains~$I_n\setminus\{a\}$, the case of~$I_n\setminus\{b\}$ is symmetric.

Let~$a^-=a-1\bmod n$ and~$a^+=a+1\bmod n$.

We apply Theorem \ref{thm_main_tool} twice. First, with~$i=a^-$ and~$j=b$, there is a partial input~$\alpha$ such that~$f_b(\alpha)=0$ and
\[
\dom(\alpha)=\up[f]{\int{a,b}}\cup \up[f]{\int{b,a^-}}=\up[f]{\int{b,a^-}}
\]
as~$K$ has no simplex containing both~$a$ and~$b$. Next, with~$i=a$ and~$j=b$, there is a partial input~$\beta$ such that~$f_b(\beta)=1$ and
\[
\dom(\beta)=\up[f]{\int{a^+,b}}\cup \up[f]{\int{b,a}}=\up[f]{\int{a^+,b}}.
\]
Therefore~$\alpha$ and~$\beta$ cannot compatible, implying that~$\dom(\alpha)$ and~$\dom(\beta)$ must intersect. Their intersection is~$\up[f]{\int{a^+,a^-}}$ so~$\int{a^+,a^-}=I_n\setminus \{a\}$ belongs to~$K$. By symmetry,~$K$ also contains~$I_n\setminus \{b\}$, showing that~$K_{a,b}$ is minimal.

Now, if a minimal complex generates~$\Mon_n$ and contains at least two~$(n-1)$-intervals, then it contains a member~$K_{a,b}$ of the~$K_2$ family, so it equals~$K_{a,b}$ by minimality.
\end{proof}

\subsection{The \texorpdfstring{$K_5$}{K5} family}\label{sec_K5_family}
We now identify the generating complexes that contain exactly one~$(n-1)$-interval. They only exist for~$n\geq 5$, and are induced by~$K_5$ under vertex insertions.

\begin{definition}
The \textbf{$K_5$ family} is made of the following complexes and their circular permutations. For~$n\geq 5$ and~$1<i<j<n-1$, let
\begin{equation*}
K^n_{i,j}=\comp{\int{1,n-1},\int{0,j-1},\int{j+1,i-1},\int{i+1,0}}.
\end{equation*}
\end{definition}
An illustration is given in Figure \ref{fig_Knij}.
\begin{figure}[ht]
\centering
\includegraphics{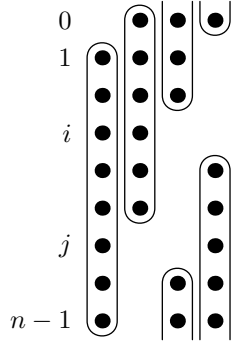}
\caption{The complex~$K^n_{i,j}$}\label{fig_Knij}
\end{figure}

\begin{remark}\label{rmk_Knij}
There is no redundancy in this description of the family: one easily checks that for each~$n$, if~$(i,j)\neq (i',j')$, then the complexes~$K^n_{i,j}$ are distinct and even incomparable w.r.t.~inclusion.
\end{remark}

We first show that this family deserves its name.
\begin{proposition}\label{prop_Knij}
Every~$K^n_{i,j}$ can be obtained from~$K_5$ by vertex insertions and circular permutations.
\end{proposition}
\begin{proof}
For~$n=5$, the only possible values of~$i$ and~$j$ are~$i=2$ and~$j=3$. The complex~$K^5_{2,3}=\comp{\int{1,4},\int{0,2},\int{4,1},\int{3,0}}$ is the image of~$K_5=\comp{\int{3,1},\int{2,4},\int{1,3},\int{0,2}}$ by the circular permutation~$x\mapsto x+3\bmod 5$.
 
We now show that every~$K^n_{i,j}$ can be obtained from~$K^5_{2,3}$ by successive insertions.
\begin{figure}[ht]
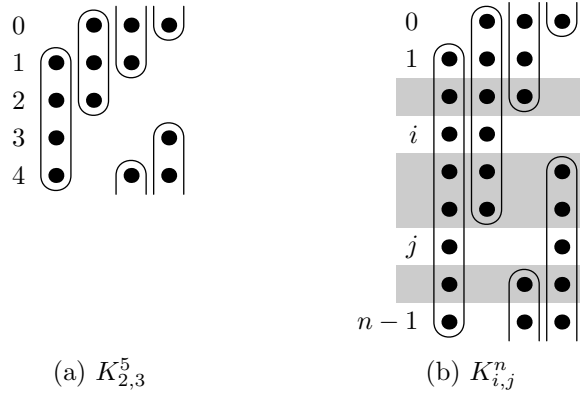

\centering
\subcaptionbox{$K^5_{2,3}$}{\includegraphics{Images/complex-10}}
\hspace{2cm}
\subcaptionbox{$K^n_{i,j}$}{\includegraphics{Images/complex-11}}
\caption{Vertex insertions from~$K^5_{2,3}$ to~$K^n_{i,j}$: the gray rows correspond to the inserted vertices}\label{fig_insertions}
\end{figure}
The reader might get convinced by staring at Figure \ref{fig_insertions}, we nonetheless give some details. Let~$1<i<j<n-1$. Inserting a vertex at certain positions in~$K^n_{i,j}$ gives rise to a complex in the same family:
\begin{itemize}
\item Position~$2$: $\ins{2}{K}=K^{n+1}_{i+1,j+1}$,
\item Position~$i+1$: $\ins{i+1}{K}=K^{n+1}_{i,j+1}$,
\item Position~$j+1$: $\ins{j+1}{K}=K^{n+1}_{i,j}$.
\end{itemize}
Therefore, if~$n\geq 6$ and~$1<i<j<n-1$, then one can start from~$K^5_{2,3}$ and successively insert vertices at suitable positions to reach~$K^n_{i,j}$:
\begin{itemize}
\item Inserting~$i-2$ vertices at position~$2$ yields~$K^{i+3}_{i,i+1}$,
\item Then inserting~$j-i-1$ vertices at position~$i+1$ yields~$K^{j+2}_{i,j}$,
\item Finally, inserting~$n-j-2$ vertices at position~$j+1$ yields~$K^{n}_{i,j}$.\qedhere
\end{itemize}
\end{proof}

\begin{theorem}[One $(n-1)$-interval]\label{thm_one_interval}
The minimal complexes generating~$\Mon_n$ and having exactly one~$(n-1)$-interval are the members of the~$K_5$ family.
\end{theorem}


The remainder of this section is devoted to the proof of this result.

First, these complexes easily generate~$\Mon_n$ because~$K_5$ generates~$\Mon_5$ (Proposition \ref{prop_K5}) and each~$K^n_{i,j}$ can be obtained, up to symmetry, by inserting vertices in~$K_5$ (Proposition \ref{prop_Knij}), so~$K^n_{i,j}$ generates~$\Mon_n$ (Proposition \ref{prop_insertion}).

We now show that every minimal complex generating~$\Mon_n$ and having one~$(n-1)$-interval is some~$K^n_{i,j}$ up to symmetry, and that every~$K^n_{i,j}$ is minimal. In the next results, we do not assume that~$n\geq 5$, which will be a consequence.

\begin{lemma}\label{lem_mono_3_intervals}
Let~$K$ generate~$\Mon_n$. For~$i,j\in I_n$,~$K$ must contain~$\int{0,j}$ or~$\int{i,0}$ or~$\int{j+1,0}\cup\int{0,i-1}$.
\end{lemma}
\begin{proof}
We assume that~$K$ contains neither~$\int{0,j}$ nor~$\int{i,0}$ and show that~$K$ contains~$\int{j+1,0}\cup\int{0,i-1}$. Let~$f$ generate~$\Mon_n$ with~$K_f\subseteq K$. We apply Theorem \ref{thm_main_tool} twice, to~$(i',j')=(i-1\bmod n,0)$ and to~$(i',j')=(j,0)$:
\begin{itemize}
\item There is a partial input~$\alpha$ such that~$f_0(\alpha)=0$, with
\[
\dom(\alpha)=\up[f]{\int{0,i-1}} \cup\up[f]{\int{i,0}}=\up[f]{\int{0,i-1}},
\]
\item There is a partial input~$\beta$ such that~$f_0(\beta)=1$, with
\[
\dom(\beta)=\up[f]{\int{0,j}} \cup\up[f]{\int{j+1,0}}=\up[f]{\int{j+1,0}}.
\]
\end{itemize}

Therefore, the domains of~$\alpha$ and~$\beta$ intersect, so~$K$ contains~$\int{0,i-1}\cup\int{j+1,0}$. 
\end{proof}

\begin{lemma}\label{lem_ij}
If~$K$ generates~$\Mon_n$ and has no~$(n-1)$-interval containing~$0$, then there exist~$i,j$ satisfying~$1<i<j<n-1$, such that~$K$ contains
\[
\{\int{0,j-1},\int{i+1,0},\int{j+1,i-1}\}.
\]
\end{lemma}
\begin{proof}
Let~$i\geq 0$ be minimal such that~$\int{i+1,0}\in K$ and~$j\leq n$ be maximal such that~$\int{0,j-1}\in K$. Both are well-defined because~$K$ contains~$\{0\}=\int{n,0}=\int{0,0}$ (otherwise, the generation procedure would always assign the same value to~$0$, but~$0$ can take both binary values in~$\Mon_n$). Moreover, one has~$i\leq n-1$ and~$j\geq 1$.

As~$K$ has no~$(n-1)$-interval containing~$0$, one has~$i>1$ and~$j<n-1$. By minimality of~$i$ and maximality of~$j$,~$K$ contains neither~$\int{0,j}$ nor~$\int{i,0}$, so by Lemma \ref{lem_mono_3_intervals},~$K$ contains~$S:=\int{j+1,0}\cup \int{0,i-1}$. If~$i>j$ then~$S=[0,n-1]$, and if~$i=j$ then~$S=[0,n-1]\setminus \{i\}$. Both cases are impossible because~$K$ has no~$(n-1)$-interval containing~$0$. Therefore,~$i<j$ hence~$S=\int{j+1,i-1}$.
\end{proof}

Therefore, if~$K$ contains one~$(n-1)$-interval then, up to some circular permutation, we can assume that it is~$\int{1,n-1}$. As~$K$ has no~$(n-1)$-interval containing~$0$,~$K$ contains~$K^n_{i,j}$ for some~$i,j$ by Lemma \ref{lem_ij}.

\begin{proposition}
Let~$1<i<j<n-1$. $K^n_{i,j}$ is minimal generating~$\Mon_n$.
\end{proposition}
\begin{proof}
Let~$K'\subseteq K^n_{i,j}$ generate~$\Mon_n$. We first show that~$K'$ must contain~$\int{1,n-1}$. Lemma \ref{lem_technical} implies that~$K'$ contains one of the following intervals:
\[
\int{0,j},\int{j+1,i},\int{i+1,1},\int{i,0},\int{1,n-1}.
\]
The only one that belongs to~$K^n_{i,j}$ is~$\int{1,n-1}$, which therefore must belong to~$K'$.

As~$\int{1,n-1}$ is the only~$(n-1)$-interval of~$K'$, Lemma \ref{lem_ij} implies that~$K'$ contains~$K^n_{i',j'}$ for some~$i',j'$ satisfying~$1<i'<j'<n-1$. As a result,~$K^n_{i',j'}\subseteq K'\subseteq K^n_{i,j}$, implying~$i'=i$,~$j'=j$ and~$K'=K^n_{i,j}$ by Remark \ref{rmk_Knij}.
%
\end{proof}

\begin{remark}[$n\leq 4$]
The conclusion of Lemma \ref{lem_ij}, i.e.~the existence of~$i,j$ satisfying~$1<i<j<n-1$, is possible only when~$n\geq 5$. Therefore, if~$K$ is a complex generating~$\Mon_n$ for~$n\leq 4$, then every element of~$I_n$ belongs to an~$(n-1)$-interval. It implies that for~$n\in\{2,3,4\}$, the minimal complexes generating~$\Mon_n$ are the complexes in the~$K_2$ family.
\end{remark}

\subsection{The \texorpdfstring{$K_8$}{K8} family}\label{sec_K8_family}

In Section \ref{sec_Mon8}, we show that the complex
\[
K_8=\comp{\int{0,5},\int{2,7},\int{4,1},\int{6,3}}
\]
generates~$\Mon_8$. It induces a whole family of minimal complexes generating~$\Mon_n$ for~$n\geq 8$, obtained by inserting vertices in~$K_8$.

\begin{definition}
The \textbf{$K_8$ family} is made of the complexes
\[
\comp{\int{a_3,a_2-1},\int{a_2,a_1-1},\int{a_1,a_0-1},\int{a_0,a_3-1}},
\]
where~$n\geq 8$ and~$a_0,a_1,a_2,a_3\in\Z$ are such that~$a_{i+1}-a_i\geq 2$ and~$a_3-a_0\leq n-2$.
\end{definition}

\begin{proposition}\label{prop_K8_family}
The members of the~$K_8$ family can be obtained from~$K_8$ by vertex insertions and circular permutations.
\end{proposition}
\begin{proof}
For~$n=8$, the only possibility is~$a_{i+1}=a_i+2$, so the complex is~$K_8$ or its image by the circular permutation~$a\mapsto a+1\bmod 8$.

Let now~$n>8$ and~$K$ be part of the~$K_8$ family. Applying a circular permutation, we can assume that~$a_3=n$, so that
\[
K=\comp{\int{0,a_2-1},\int{a_2,a_1-1},\int{a_1,a_0-1},\int{a_0,n-1}}.
\]
This complex can be obtained from~$K_8$ by inserting vertices: $a_0-2$ insertions at position~$1$, $a_1-a_0-2$ insertions at position~$3$,~$a_2-a_1-2$ insertions at position~$5$ and~$n-a_2-2$ insertions at position~$7$ (more correctly, the insertions should be made at the new positions of~$3,5,7$).
\end{proof}

\begin{theorem}\label{thm_from_eight}
The members of the~$K_8$ family are minimal generating~$\Mon_n$.
\end{theorem}
\begin{proof}
We know by Proposition \ref{prop_K8} that~$K_8$ generates~$\Mon_8$, we show that it is minimal.
%
It is more convenient to first apply the circular permutation~$a\mapsto a+1\bmod 8$, yielding
 \[
 K'_8=\comp{\int{1,6},\int{7,4},\int{5,2},\int{3,0}}.
 \]
 Let~$K'\subseteq K'_8$ generate~$\Mon_8$. We apply Lemma \ref{lem_technical} to~$i=3$ and~$j=5$, implying that~$K'_8$ contains one of the following intervals:
\[
\int{0,5},\int{6,3},\int{4,1},\int{3,0},\int{1,7}.
\]
The only interval that belongs to~$K'_8$ is~$\int{3,0}$, which therefore must belong to~$K'$. By symmetry,~$K'$ contains every maximal interval of~$K'_8$, hence~$K'=K'_8$. More precisely, observe that~$K'_8$ is very symmetric: it is invariant by the circular permutation~$a\mapsto a+2\bmod 8$. Any maximal interval of~$K'_8$ can be sent to~$\int{3,0}$ by iterating this permutation, hence must belong to~$K'$. We have proved the minimality of~$K'_8$, hence of~$K_8$.


For~$n>8$, up to a circular permutation, a complex in the~$K_8$ family is obtained by vertex insertions, as explained in Proposition \ref{prop_K8_family}. Moreover, in~$K_8$, each one of the pairs~$(0,1)$, $(2,3)$, $(4,5)$ and~$(6,7)$ has the particular property that every interval of~$K_8$ either contains this pair, or is disjoint from it. Therefore, Corollary \ref{cor_insertion_minimal} implies that the resulting complex~$K$ is minimal generating~$\Mon_n$.
\end{proof}

\subsection{Small values of \texorpdfstring{$n$}{n}}
For small values of~$n$, we can almost completely describe the minimal complexes generating~$\Mon_n$. There is nonetheless one exception: for~$n=6$, there is one complex for which we do not know whether it generates~$\Mon_6$. It is the complex made of all the~$4$-intervals,
\[
K_6=\comp{\int{0,3},\int{1,4},\int{2,5},\int{3,0},\int{4,1},\int{5,2}},
\]
and is illustrated in Figure \ref{fig_K6}.
\begin{figure}[ht]
\centering
\includegraphics{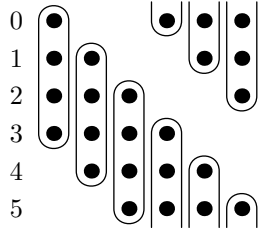}
\caption{The complex~$K_6$}\label{fig_K6}
\end{figure}
We define the~$K_6$ family as the family of complexes obtained from~$K_6$ by vertex insertions.

\begin{theorem}\label{thm_small_values}
For~$2\leq n\leq 7$, the minimal complexes generating~$\Mon_n$ belong to the families of~$K_2$,~$K_5$,~$K_6$ or~$K_7$.
\end{theorem}

Conversely, we have previously seen that all the members of the families of~$K_2$,~$K_5$ and~$K_7$ are minimal generating~$\Mon_n$, and we leave open whether~$K_6$ indeed generates~$\Mon_6$. A computer program shows that the inference system from Section \ref{sec_rule_system} does not find any conflict on~$K_6$ and~$\Mon_6$. The inference system is incomplete for general languages, and we do not know whether it is complete for~$\Mon_n$. A computer search shows that there is no function~$f:\{0,1\}^6\to \{0,1\}^6$ generating~$\Mon_6$ with~$K_f\subseteq K_6$ and that commutes with negation.

The proof of Theorem \ref{thm_small_values} is given in the appendix, Section \ref{app_small}.

\section{The \texorpdfstring{$3/4$}{3/4} ratio}\label{sec_34}
In this section, we investigate the following question: how small can the intervals of a complex generating~$\Mon_n$ be? More precisely, what is the minimal value of~$k$ such that there is a complex generating~$\Mon_n$ and whose intervals all have size at most~$k$? We show that~$k$ is approximately~$\frac{3n}{4}$.

\begin{definition}\label{def_mu}
Let~$\mu(n)$ be the minimal~$k$ such that there exists a complex generating~$\Mon_n$ whose intervals all have sizes at most~$k$.
\end{definition}
\begin{theorem}[Complexes with shortest intervals]\label{thm_mu}
For all~$n$, one has
\[
\bigfloor{\frac{3n+1}{4}}\leq \mu(n)\leq \bigceil{\frac{3n}{4}}.
\]
\end{theorem}
When~$n=0$ or~$1\bmod 4$, the lower and upper bound coincide, so they give the exact value of~$\mu(n)$. In the other cases, they give the value~$\mu(n)$ up to~$1$. We give the first values in Table \ref{tab_mu}.
\begin{table}[ht]
\[
\begin{array}{|c||ccc|}
\hline
\rule[-2.5mm]{0pt}{4.3ex}n&\floor{\frac{3n+1}{4}}&\mu(n)&\ceil{\frac{3n}{4}}\\
\hline
\hline
\rule{0pt}{2.5ex}
1&1&1&1\\\hline
2&1&1&2\\\hline
3&2&2&3\\\hline
4&3&3&3\\\hline
5&4&4&4\\\hline
6&4&?&5\\\hline
7&5&5&6\\\hline
8&6&6&6\\\hline
\end{array}
\]
\caption{First values of~$\mu(n)$ and its bounds}\label{tab_mu}
\end{table}

The upper bound is not optimal for~$n=2,3$ and~$7$, because~$\Mon_2$ and~$\Mon_3$ are generated by members of the~$K_2$ family, and~$\Mon_7$ is generated by~$K_7$. For~$n=6$, we do not know the exact value of~$k$, which is~$4$ if~$K_6$ generates~$\Mon_6$ and~$5$ otherwise.

\subsection{Upper bound}\label{sec_upper_bound}
The upper bound can be directly obtained from the~$K_8$ family of complexes presented in Section \ref{sec_K8_family}.

\begin{corollary}[Upper bound]
For~$n\geq 8$, there exists a complex generating~$\Mon_n$ whose intervals all have size at most~$\ceil{\frac{3n}{4}}$.
\end{corollary}
\begin{proof}
Let~$k=\ceil{\frac{3n}{4}}$ and~$a_i=(i+1)(n-k)$.  One can check that~$n\geq 8$ implies~$n-k\geq 2$, so~$a_{i+1}-a_i\geq 2$ and~$n-a_2\geq 2$. Therefore, we can apply Theorem \ref{thm_from_eight}: the complex
\[
K=\{\int{a_3,a_2-1},\int{a_2,a_1-1},\int{a_1,a_0-1},\int{a_0,a_3-1}\}
\]
generates~$\Mon_n$. Its first three maximal intervals have length~$k$ and its fourth one has length~$3(n-k)\leq k$.

We note that this choice of~$a_i$ is optimal: the sum of the lengths of the four maximal intervals in~$K$ is~$3n$, so one of them must have length at least~$\ceil{\frac{3n}{4}}$ (otherwise, the sum of the lengths is at most~$4(\ceil{\frac{3n}{4}}-1)<3n$).
\end{proof}

\subsection{Lower bound}
We now present the proof of the lower bound, which is unfortunately very technical. It has been found by automatically searching conflicts for small values of~$n$ and complexes made of small intervals, and then identifying common patterns in these conflicts. It would be interesting to find a more elegant argument that also explains the~$3/4$ ratio.
\begin{theorem}\label{thm_lower_bound}
If a complex generates~$\Mon_n$, then it contains an interval of length~$\floor{\frac{3n+1}{4}}$.
\end{theorem}
\begin{proof}
Let~$k=\floor{\frac{3n+1}{4}}$. 
For~$n\leq 6$, the statement follows from Theorem \ref{thm_small_values}. We now assume that~$n\geq 7$. We will use the facts that~$k<n$ and~$3k\geq 2n+1$ (the latter holds for~$n\geq 9$ because~$3k>3(\frac{3n+1}{4}-1)=2n+\frac{n-9}{4}\geq 2n$, and can be checked for~$n=7$ or~$8$ in Table \ref{tab_mu}).

Assume for a contradiction that~$K$ generates~$\Mon_n$ and contains no interval of size~$k$. The argument consists in using the inference system from Section \ref{sec_rule_system} to build more and more constraints that eventually lead to a contradiction.

\newcommand{\aaa}{a}
\newcommand{\bbb}{b}
\newcommand{\ccc}{c}
\begin{lemma}\label{lem_b}
Assume that~$f$ generates~$\Mon_n$ and~$K_f$ contains no~$k$-interval. Let~$\bbb=2(n-k)$. For~$2k-n-1\leq j\leq k-1$, there exists~$\gamma$ such that~$f_0(\gamma)=0$ and
\[
\dom(\gamma)=\up[f]{\int{0,j}}\cup\up[f]{\int{j+1,\bbb+j+1}}.
\]
\end{lemma}
\begin{proof}
To lighten the notations, we drop the subscript~$f$ in~$\W_f$. We prove the result by induction on~$j$. Note that by symmetry, the statement implies more generally that for any~$a\in I_n$, there exists~$\gamma$ such that~$f_a(\gamma)=0$ and~$\dom(\gamma)=\up{\int{a,a+j}}\cup\up{\int{a+j+1,a+b+j+1}}$. This more general statement will be used as induction hypothesis.

For~$j=2k-n-1$, the result is a direct application of Theorem \ref{thm_main_tool}, giving~$\dom(\gamma)=\up{\int{0,j}}\cup\up{\int{j+1,n}}$, and indeed~$b+j+1=n$ (and the version for any~$a$ holds by symmetry).

We assume the result for~$j$ satisfying~$2k-n-1\leq j<k-1$ and prove it for~$j+1$. Let~$\aaa=n-k$ and
\begin{align*}
\I_0&=\int{0,j}&\I_1&=\int{j+1,\bbb+j+1}\\
\J_0&=\int{\aaa,\aaa+j}&\J_1&=\int{\aaa+j+1,\aaa+\bbb+j+1}.
\end{align*}

From the induction hypothesis, there exist partial inputs~$\alpha$ and~$\beta$ such that:
\begin{itemize}
\item $f_0(\alpha)=0$ and~$\dom(\alpha)=\up{\I_0}\cup\up{\I_1}$,
\item $f_\aaa(\beta)=0$ on~$\dom(\beta)=\up{\J_0}\cup\up{\J_1}$.
\end{itemize}

\begin{claim}
The domains of~$\alpha$ and~$\beta$ are disjoint.
\end{claim}
The intersection of these domains is
\[
\dom(\alpha)\cap\dom(\beta)=\up{\I_0\cup\J_0}\cup\up{\J_0\cup\I_1}\cup\up{\I_1\cup\J_1}\cup\up{\J_1\cup\I_0}.
\]
We prove that it is empty by showing that each union~$\I_u\cup\J_v$ is an interval of length at least~$k$.

\begin{figure}[ht]
\centering
\includegraphics{Images/intervals-2}
\caption{$\I_0\bef\J_0\bef\I_1\bef\J_1\bef\I_0$}\label{fig_overlap}
\end{figure}

We check that~$\I_0\bef\J_0\bef\I_1\bef\J_1\bef\I_0$ (illustrated in Figure \ref{fig_overlap}). Indeed, each pair of consecutive intervals can be expressed as~$\int{x,z}\bef\int {y,t}$ where~$x\leq y\leq z\leq t<x+n$ (Lemma \ref{lem_bef}), as summarized by the following diagram (see paragraph after Lemma \ref{lem_bef}):
\begin{center}
\begin{tikzcd}
0\arrow[r]& \wid[a+j]{a}\arrow[r]\dlArrow& \wid[b+j+1]{j+1}\arrow[r]\dlArrow& \wid[a+b+j+1]{a+j+1}\arrow[r]\dlArrow& \wid[n+j.]{n}\dlArrow\\
j\arrow[r]& a+j\arrow[r]\ulArrow& b+j+1\arrow[r]\ulArrow& a+b+j+1\arrow[r]\ulArrow& n+j.\ulArrow
\end{tikzcd}
\end{center}

%

These inequalities all follow from:
\begin{gather*}
k<n,\\
2k-n-1\leq j<k-1,\\
2n+1\leq 3k.
\end{gather*}

Therefore, the consecutive unions of intervals are intervals:
\begin{itemize}
\item $\I_0\cup \J_0=\int{0,\aaa+j}$ has size~$\aaa+j+1\in [k,n]$,
\item $\J_0\cup\I_1=\int{\aaa,\bbb+j+1}$ has size~$\bbb+j+2-\aaa\in [k,n]$,
\item $\I_1\cup\J_1=\int{j+1,\aaa+\bbb+j+1}$ has size~$\aaa+\bbb+1\in [k,n]$,
\item $\J_1\cup \I_0=\int{\aaa+j+1,n+j}$ has size~$n-\aaa=k\in [k,n]$.
\end{itemize}

As a result,~$K$ does not contain any~$\I_u\cup\J_v$, which means that~$\up{\I_u\cup\J_v}$ is empty. Therefore,~$\alpha$ and~$\beta$ have disjoint domains and the claim is proved.

This claim implies that~$\alpha$ and~$\beta$ are compatible, so we can apply Lemma \ref{lem_mono}. Let~$\ccc=n-1\in \int{\aaa,0}$. One has~$f_\ccc(\gamma)=0$ where
\[
\gamma=\alpha\res{\up{\int{\ccc,0}}}\cup \beta\res{\up{\int{\aaa,\ccc}}}.
\]

The interval~$\int{\aaa,\ccc}$ has size~$\ccc+1-\aaa=n-\aaa=k$ so~$\up{\int{\aaa,\ccc}}$ is empty. Therefore,~$\gamma=\alpha\res{\up{\int{\ccc,0}}}$ has domain
\begin{align*}
\dom(\gamma)&=\dom(\alpha)\cap\up{\int{\ccc,0}}\\
&=\big(\up{\int{0,j}}\cup\up{\int{j+1,b+j+1}}\big)\cap \up{\int{\ccc,0}}\\
&=\up{\int{n-1,j}}\cup\up{\int{j+1,\bbb+j+1}}.
\end{align*}

By symmetry, we can apply the circular permutation~$x\mapsto x+1\bmod n$, implying that there exists a partial input~$\delta$ such that~$f_0(\delta)=0$ and~$\dom(\delta)=\up{\int{0,j+1}}\cup\up{\int{j+2,\bbb+j+2}}$ which proves the induction step.
\end{proof}
Let
\[
\I=\int{k,2n-k}.
\]
We apply Lemma \ref{lem_b} to~$j=k-1$, giving~$\gamma$ such that~$f_0(\gamma)=0$ and~$\dom(\gamma)=\up{\int{0,k-1}}\cup \up{\int{k,2n-k}}$. As~$\int{0,k-1}$ has size~$k$,~$\up{\int{0,k-1}}$ is empty so~$\dom(\gamma)=\up{\I}$.

Let
\begin{align*}
\J_0&=\int{2k-n,0}\\
\J_1&=\int{0,2k-n-1}=\int{0,2k-1}
\end{align*}
\begin{figure}[ht]
\centering
\includegraphics{Images/intervals-4}
\caption{$\J_0\bef\I\bef\J_1$}
\end{figure}

By Theorem \ref{thm_main_tool}, there exists~$\alpha$ such that~$f_0(\alpha)=1$ and~$\dom(\alpha)=\up{\int{2k-n,n}}\cup \up{\int{0,2k-n-1}}=\up{\J_0}\cup\up{\J_1}$.

We show that~$\alpha$ and~$\gamma$ have disjoint domains. One has~$\J_0\bef\I\bef\J_1$, because the endpoints of these intervals satisfy the inequalities summarized by the following diagram,
\begin{center}
\begin{tikzcd}
2k-n\arrow[r]& \wid[2n-k]{k}\arrow[r]\dlArrow& \wid[2k-1]{n}\dlArrow\\
\wid[2k-n]{n}\arrow[r]& 2n-k\arrow[r]\ulArrow& 2k-1\ulArrow,
\end{tikzcd}
\end{center}
which all follow from~$k\leq n$ and~$3k\geq 2n+1$. Therefore,
\begin{itemize}
\item $\J_0\cup\I=\int{2k-n,2n-k}$ has size~$3(n-k)+1\in [k,n]$,
\item $\I\cup \J_1=\int{k,2k-1}$ has size~$k$.
\end{itemize}
As a result,~$\alpha$ and~$\gamma$ have disjoint domains so they are compatible, but they give opposite values to~$0$, which is a contradiction.
\end{proof}

\begin{remark}
One might be tempted to prove Theorem \ref{thm_lower_bound} by induction on~$n$ in the following way. Let~$k_n=\floor{\frac{3n+1}{4}}$ and assume for a contradiction that the complex~$K$ made of all the~$(k_n-1)$-intervals generates~$\Mon_{n}$. Delete~$4$ vertices regularly spaced, so that at least~$3$ vertices are removed from each maximal interval of~$K$. The resulting complex~$K'$ has~$n-4$ vertices and intervals of length at most~$k_n-4=\floor{\frac{3(n-1)}{4}}$. When~$n=2\bmod 4$, this quantity is strictly smaller than~$k_{n-1}$ so we obtain a contradiction by induction hypothesis. 

However, this argument does not work: there is no way to choose~$4$ elements in~$[0,n-1]$ so that every~$(k_n-1)$-interval contains at least~$3$ of them. The only way to make it work would be to slightly increase~$k_n$ to~$\ceil{\frac{3n}{4}}+2$, but then the base case would not hold (and~$k_n$ would exceed the upper bound~$\ceil{\frac{3n}{4}}$, so the statement would be wrong anyway).
\end{remark}

\section{Future directions}\label{sec_future}

\subsection{Language generation via combinatorial topology}

As explained in \cite{H25}, the language generation problem can be reformulated in the framework developed in \cite{HS99,HKR13}, which enables the analysis of distributed algorithms using techniques from combinatorial topology.

In particular, to each language~$L\subseteq A^I$ one can associate a chromatic labeled simplicial complex~$\O_L$, whose simplices represent the elements of~$L$ and in which simplices sharing many vertices correspond to strings have many common values. Precisely, the vertices of~$\O_L$ are the pairs~$(i,x_i)\in I\times A$ for~$i\in I$ and~$x\in L$, and each string~$x\in L$ gives rise to a simplex~$\{(i,x_i):i\in I\}$. Therefore,~$\O_L$ is a pure simplicial complex of dimension~$|I|-1$, i.e.~its maximal simplices all have~$|I|$ vertices. It is chromatic and labeled because for each vertex~$(i,a)$, its component~$i$ is seen as a color and its component~$a$ is seen as a label, and the vertices of each simplex all have distinct colors.

For~$n\geq 2$, the complex~$\O_{\Mon_n}$, seen as a topological space, is the product of the circle with the~$(n-2)$-dimensional ball and is shown in Figures \ref{fig_chromatic2}, \ref{fig_chromatic3} and \ref{fig_chromatic4} for~$n=2,3,4$. The colors of the vertices are elements of~$\{0,1,2,3\}$ and are visualized as $\{\pic{chromatic-100},\pic{chromatic-101},\pic{chromatic-102},\pic{chromatic-103}\}$, and their labels belong to~$\{0,1\}$ (in the pictures, the colors of the simplices have no meaning and are there only to help distinguishing them).

\begin{figure}[ht]
\centering
\includegraphics{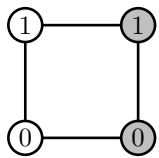}
\caption{The output complex of~$\Mon_2$ is a circle}\label{fig_chromatic2}
\end{figure}

\begin{figure}[ht]
\centering
\subcaptionbox{A flattened version: the left and right edges are glued together}{\includegraphics{Images/chromatic-mono-1}}\hspace{1cm}
\subcaptionbox{The complex}[.2\textwidth]{\includegraphics[width=.15\textwidth]{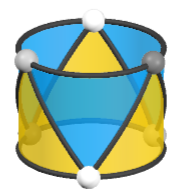}}
\caption{The output complex of~$\Mon_3$ is a cylinder}\label{fig_chromatic3}
\end{figure}

\begin{figure}[ht]
\centering
\subcaptionbox{An unfold version: the left and right triangles are glued together.}{\raisebox{2mm}{\includegraphics[width=.45\textwidth]{
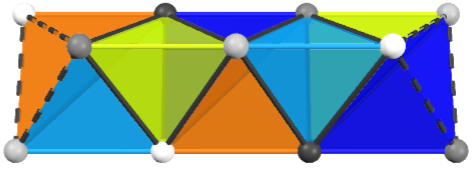
}}}\hspace{1cm}
\subcaptionbox{The surface of the complex}[.4\textwidth]{\includegraphics[width=.35\textwidth]{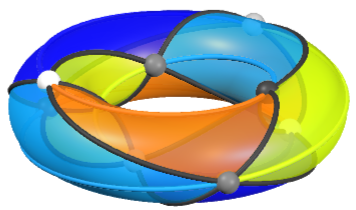}}
\caption{The output complex of~$\Mon_4$ is a solid doughnut (the models can be interactively visualized at \cite{DesmosMon4Unfold} and \cite{DesmosMon4})}\label{fig_chromatic4}
\end{figure}

Each communication complex~$K$ over~$I$, together with an input alphabet~$B$, can be turned into a similar complex~$\II_K(B)$ in which the simplices represent the inputs in~$B^J$, the vertices are also assigned colors~$i\in I$ and their labels are the restrictions of the inputs to~$\W(i)$. Whether~$K$ generates a language~$L$ is then equivalent to the existence of a surjective simplicial map~$f:\II_K(B)\to \O_L$ that is color-preserving.

It would be interesting to understand whether this reformulation can help to study the problem of local generation of~$\Mon_n$, using tools from combinatorial topology.

\subsection{Open questions}
The main open problem is a complete classification of the minimal complexes generating~$\Mon_n$, for all~$n$. Many complexes are obtained by inserting vertices in smaller complexes, but we have seen that for some values of~$n$, namely~$n=2,5,7,8$, a new way of generating~$\Mon_n$ appears.
%

We raise the following questions:
\begin{itemize}
\item Are there infinitely many values of~$n$ coming with a new way of generating~$\Mon_n$? Or is there a finite family of complexes that induce all of them by insertions?
\item Does~$K_6$ generate~$\Mon_6$?
\item What is the exact value of~$\mu(n)$ (Definition \ref{def_mu}), and is there an elegant argument explaining its value?
\end{itemize}


The proofs of Lemma \ref{lem_technical} and Theorem \ref{thm_lower_bound} are ``empirical'' in the sense that they were obtained with the help of computer search  (note however that the process was far from automatic, and the statements presented in the article are the results of an effort to unify and generalize arguments found for particular complexes and small values of~$n$). These arguments cannot really be shortened, because they are essentially derivations of conflicts with smallest derivation trees, in the inference system from Section \ref{sec_rule_system}. Still, one could hope for more theoretically grounded arguments, abstracting away the most tedious details.

More generally, is the behavior of~$\Mon_n$ inherently chaotic, or is it possible to have a more structural understanding of the local generation aspects of this language?

\newpage
\appendix
\section{Appendix}\label{sec_app}

Some of the proofs are given in this section.

\subsection{Proof of Lemma \ref{lem_technical}}\label{app_lem}
We first recall the statement of Lemma \ref{lem_technical}.
 
\begin{lemma*}
Let~$1<i<j<n-1$. If a complex~$K$ generates~$\Mon_n$, then~$K$ contains at least one of the following intervals:
\[
\int{0,j},\int{j+1,i},\int{i+1,1},\int{i,0},\int{1,n-1}.
\]
\end{lemma*}
\begin{proof}
Let~$f$ generate~$\Mon_n$ with~$K_f\subseteq K$. Again, we write~$\W$ for~$\W_f$. We assume for a contradiction that~$K$, hence~$K_f$, contains none of these intervals.
\begin{claim}
There exists a partial input~$\gamma$ such that~$\dom(\gamma)=\up{\int{1,j}}$ and~$f_i(\gamma)=0$.
\end{claim}
Applying Theorem \ref{thm_main_tool} twice, there exist partial inputs~$\alpha,\beta$ such that:
\begin{itemize}
\item $f_0(\alpha)=1$ and~$\dom(\alpha)=\up{\int{i+1,0}}\cup\up{\int{0,i}}$,
\item $f_1(\beta)=0$ and~$\dom(\beta)=\up{\int{j+1,1}}\cup\up{\int{1,j}}$.
\end{itemize}

One has
\[
\int{0,i}\bef\int{1,j}\bef\int{i+1,0}\bef\int{j+1,1}\bef\int{0,i}
\]
because the endpoints of these intervals have representatives in~$\Z$ satisfying the inequalities expressed in Lemma \ref{lem_bef} and summarized in the following diagram:
\begin{center}
\begin{tikzcd}
0\arrow[r]& 1\arrow[r]\dlArrow& i+1\arrow[r]\dlArrow& j+1 \arrow[r]\dlArrow& \wid[n+i.]{n}\dlArrow\\
i\arrow[r]& j\arrow[r]\ulArrow& \wid[i+1]{n}\arrow[r]\ulArrow& \wid[j+1]{n+1}\arrow[r]\ulArrow& n+i.\ulArrow
\end{tikzcd}
\end{center}

Therefore,  the consecutive unions of these intervals are
\[
\int{0,j},\int{1,0},\int{i+1,1}\text{ and }\int{j+1,i}
\]
which are all missing in~$K$. As a result,~$\alpha$ and~$\beta$ have disjoint domains so they are compatible. Let~$\gamma=\restr{\alpha}{\up{\int{i,0}}}\sqcup\restr{\beta}{\up{\int{1,i}}}=\restr{\beta}{\up{\int{1,i}}}$. By Lemma \ref{lem_mono}, one has~$f_i(\gamma)=0$ and
\begin{align*}
\dom(\gamma)&=\dom(\beta)\cap\up{\int{1,i}}\\
&=\up{\int{j+1,i}}\cup\up{\int{1,j}}\\
&=\up{\int{1,j}}
\end{align*} and the claim is proved.

On the other hand, by Theorem \ref{thm_main_tool} there is a partial input~$\delta$ such that~$f_i(\delta)=1$ and~$\dom(\delta)=\up{\int{0,i}}\cup\up{\int{i,n-1}}$.

Finally, one has
\begin{align*}
\dom(\gamma)\cap\dom(\delta)&=\up{\int{1,j}}\cap\big(\up{\int{0,i}}\cup\up{\int{i,n-1}}\big)\\
&=\up{\int{0,j}}\cup\up{\int{1,n-1}}\\
&=\emptyset,
\end{align*}
so~$\gamma$ and~$\delta$ are compatible but give opposite values to~$i$, which is a contradiction.
\end{proof}

\subsection{Proof of Theorem \ref{thm_small_values}}\label{app_small}

We recall the statement of Theorem \ref{thm_small_values}.
\begin{theorem*}
For~$2\leq n\leq 7$, the minimal complexes generating~$\Mon_n$ belong to the family of~$K_2$,~$K_5$,~$K_6$ or~$K_7$.
\end{theorem*}

\paragraph{The case \texorpdfstring{$n=5$}{n=5}.}
\begin{proposition}
The minimal complexes generating~$\Mon_5$ are, up to symmetry:
\begin{itemize}
\item In the~$K_2$ family,
\item $K_5$.
\end{itemize}
\end{proposition}
\begin{proof}
We show that if~$K$ generates~$\Mon_5$, then~$K$ contains a~$4$-interval. Applying Lemma \ref{lem_technical} to~$i=2$ and~$j=3$,~$K$ contains at least of these intervals:
\[
\int{0,3},\int{4,2},\int{3,1},\int{2,0},\int{1,4},
\]
which are all the~$4$-intervals. If~$K$ contains at least two~$4$-intervals, then it belongs to the~$K_2$ family by Proposition \ref{prop_mono_2_simplices}. If it contains only one~$4$-interval, then it belongs to the~$K_5$ family by Theorem \ref{thm_one_interval}. For~$n=5$, the only members of this family are~$K_5$ and its symmetric versions.
\end{proof}

\paragraph{The case \texorpdfstring{$n=6$}{n=6}.}

\begin{proposition}
If~$K$ is a minimal complex generating~$\Mon_6$, then
\begin{itemize}
\item $K$ belongs to the~$K_2$ or the~$K_5$ families,
\item Or~$K=K_6$.
\end{itemize}
\end{proposition}
\begin{proof}
If~$K$ contains at least one~$5$-interval, then it belongs to the~$K_2$ family or the~$K_5$ family by Proposition \ref{prop_mono_2_simplices} and Theorem \ref{thm_one_interval}. Now assume that~$K$ contains no~$5$-interval. Applying Lemma \ref{lem_technical} to~$i=3$ and~$j=4$,~$K$ contains at least of these intervals:
\[
\int{0,4},\int{5,3},\int{3,1},\int{3,0},\int{1,5}.
\]
All these intervals except~$\int{3,0}$ have length~$5$, so~$K$ contains~$\int{3,0}$, which has length~$4$. By symmetry,~$K$ contains every~$4$-interval, so~$K=K_6$.
\end{proof}
As previously mentioned, we do not know whether~$K_6$ generates~$\Mon_6$. 

\paragraph{The case \texorpdfstring{$n=7$}{n=7}.}
Let~$K$ generate~$\Mon_7$. Once again, if~$K$ contains at least one~$6$-interval, then~$K$ belongs to the~$K_2$ or~$K_5$ family. Assume that~$K$ contains no~$6$-interval.
\begin{proposition}\label{prop_every4}
If a complex generates~$\Mon_7$ and contains no~$6$-interval, then it contains all the~$4$-intervals.
\end{proposition}
\begin{proof}
Applying Lemma \ref{lem_technical} to~$i=2$ and~$j=5$,~$K$ contains at least one of these intervals:
\[
\int{0,5},\int{6,2},\int{3,1},\int{2,0},\int{1,6}.
\]
All these intervals except~$\int{6,2}$ have lengths at least~$6$, so~$K$ contains~$\int{6,2}$, which is a~$4$-interval. By symmetry,~$K$ contains every~$4$-interval.´
\end{proof}

Let~$A$ be the set of starting points of~$5$-intervals in~$K$.
\begin{proposition}\label{prop_A}
For each~$a\in\Z/7\Z$,~$A$ intersects~$\{a,a+2\}$ and~$\{a,a+1,a+4\}$.
\end{proposition}
\begin{proof}
Applying Lemma \ref{lem_technical} to~$i=2$ and~$j=4$,~$K$ contains at least one of these intervals:
\[
\int{0,4},\int{5,2},\int{3,1},\int{2,0},\int{1,6}.
\]
The first two intervals have length~$5$, the others have length~$6$. Therefore,~$K$ contains~$\int{0,4}$ or~$\int{5,2}$ hence~$A$ intersects~$\{5,0\}$, showing the case~$a=5$. The other cases hold by symmetry.

Applying Lemma \ref{lem_technical} to~$i=3$ and~$j=4$,~$K$ contains at least one of these intervals:
\[
\int{0,4},\int{5,3},\int{4,1},\int{3,0},\int{1,6}.
\]
The second and last intervals have length~$6$, so~$K$ contains one of the others, i.e.~$A$ intersects~$\{3,4,0\}$. It proves the case~$a=3$, the other cases hold by symmetry.
\end{proof}

We now consider two cases.

First assume that~$A$ is disjoint from~$\{a,a+1\}$ for some~$a$. By symmetry, we can assume that~$a=0$, i.e.~$A$ is disjoint from~$\{0,1\}$. Proposition \ref{prop_A} implies that~$A$ contains~$\{2,3,4,5,6\}$: indeed,~$A$ intersects~$\{0,2\}$,~$\{1,2\}$,~$\{5,0\}$, $\{6,1\}$ and~$\{0,1,4\}$. As~$K$ contains every~$4$-interval by Proposition \ref{prop_every4},~$K=\comp{\int{0,3},\int{1,4},\int{2,6},\int{3,6},\int{4,1},\int{5,2},\int{6,3}}$ which is~$\ins{6}{K_6}$, hence~$K$ belongs to the~$K_6$ family.

Now assume that~$A$ intersects~$\{a,a+1\}$ for every~$a$. As~$K$ is minimal, one has~$|A|\leq 5$ (if~$|A|\geq 6$, then~$K$ strictly contains a circular permutation of~$K_7$, which is impossible by minimality). We can assume by symmetry that~$0\notin A$. It implies that~$A$ contains~$\{1,2,5,6\}$ and~$3$ or~$4$: by our current assumption, it intersects~$\{0,1\}$ and~$\{6,0\}$, and by Proposition \ref{prop_A} it intersects~$\{0,2\}$,~$\{5,0\}$ and~$\{3,4,0\}$. Therefore,~$|A|\geq 5$ so~$A=\{1,2,3,5,6\}$ or~$\{1,2,4,5,6\}$, which both correspond to circular permutations of~$K_7$.

\end{document}